\documentclass[11pt,journal,draftcls,onecolumn,peerreviewca]{IEEEtran}
\usepackage{amsfonts}
\usepackage{amsmath}
\usepackage{amssymb}
\usepackage{bm}
\usepackage{color}
\usepackage{float}
\usepackage{graphics}
\usepackage{graphicx}
\usepackage{hyperref}
\usepackage{indentfirst}
\usepackage{mathrsfs}
\usepackage[numbers,sort&compress]{natbib}
\usepackage[amsmath,thmmarks]{ntheorem}
\usepackage{ntheorem}
\usepackage{stfloats}
\usepackage{stmaryrd}
\usepackage{soul}
\usepackage{subfigure}
\usepackage{theorem}
\usepackage{url}

\newtheorem{corollary}{Corollary}
\newtheorem{definition}{Definition}
\newtheorem{lemma}{Lemma}
\newtheorem{property}{Property}
\newtheorem{proposition}{Proposition}

\newtheorem{theorem}{Theorem}

\theoremheaderfont{\sc}\theorembodyfont{\upshape}
\theoremstyle{nonumberplain}
\theoremseparator{}
\theoremsymbol{\rule{1ex}{1ex}}
\newtheorem{proof}{Proof}

\newcommand{\RomanNumeralCaps}[1]{\MakeUppercase{\romannumeral #1}}

\par

\begin{document}
\title{On the Identifiability from Modulo Measurements under DFT Sensing Matrix}
\author{Qi Zhang, Jiang Zhu, Fengzhong Qu, Zheng Zhu, and De Wen Soh}
\maketitle
\begin{abstract}
Modulo sampling (MS) has been recently introduced to enhance the dynamic range of conventional ADCs by applying a modulo operator before sampling. This paper examines the identifiability of a measurement model where measurements are taken using a discrete Fourier transform (DFT) sensing matrix, followed by a modulo operator (modulo-DFT). Firstly, we derive a necessary and sufficient condition for the unique identification of the modulo-DFT sensing model based on the number of measurements and the indices of zero elements in the original signal. Then, we conduct a deeper analysis of three specific cases: when the number of measurements is a power of $2$, a prime number, and twice a prime number. Additionally, we investigate the identifiability of periodic bandlimited (PBL) signals under MS, which can be considered as the modulo-DFT sensing model with additional symmetric and conjugate constraints on the original signal. We also provide a necessary and sufficient condition based solely on the number of samples in one period for the unique identification of the PBL signal under MS, though with an ambiguity in the direct current (DC) component. Furthermore, we show that when the oversampling factor exceeds $3(1+1/P)$, the PBL signal can be uniquely identified with an ambiguity in the DC component, where $P$ is the number of harmonics, including the fundamental component, in the positive frequency part. Finally, we also present a recovery algorithm that estimates the original signal by solving integer linear equations, and we conduct simulations to validate our conclusions.
\end{abstract}


\section{Introduction}
In modern electronic systems, digital acquisition of analog signals is an essential operation facilitated by digital signal processors (DSPs). Analog-to-digital converters (ADCs) play a pivotal role in this process, sampling the analog signal uniformly and outputting the digital formats. For conventional ADCs, the dynamic range, defined as the ratio of the largest measurable signal to the smallest measurable signal, is $6B+1.73$ (dB), where $B$ is the bit-depth of the ADC. Due to the exponential growth in power consumption with an increase in the number of bits, the dynamic range of conventional ADCs is frequently inadequate in numerous practical scenarios, as illustrated below:
\begin{itemize}
    \item In massive multiple-input multiple-output (MIMO) communication systems \cite{zhang2018low}, low-bit ADCs are required due to the limited power supply, which introduces severe quantization noise when conventional ADCs are used. 
    \item In scenarios where strong and weak targets coexist \cite{zhang2019range}, unless a high-bit ADC is utilized, which is power-consuming, the sensor faces a dilemma. It can either prioritize the strong signal, resulting in the weak signal being submerged in quantization noise, or prioritize the weak signal, leading to clipping or saturation of the strong signal.
\end{itemize}
Recently, the unlimited sampling (US) framework has been proposed, which employs modular arithmetic before sampling to increase the dynamic range of the conventional ADC \cite{bhandari2017unlimited, bhandari2020unlimited}. The hardware implementation of US-ADC is introduced in \cite{bhandari2021unlimited}, and non-ideal folding and hysteresis issues arising from US-ADC have also been considered \cite{bhandari2021unlimited,florescu2022surprising}. 
Another line of work that uses the modulo operator to address the limitations of conventional ADCs is named mod-ADC, in which the impact of quantization noise is analyzed in detail from a rate-distortion perspective \cite{ordentlich2018modulo}. 

In modulo sampling (MS), the modulo operator is first applied to compress the dynamic range of the original signal by calculating the remainder when the signal is divided by the full-scale range of a conventional ADC, and this modulo signal is then sampled by the conventional ADC. Since the modulo operation is a many-to-one mapping and thus non-invertible, directly recovering the original signal from the modulo samples is generally impossible for arbitrary signals. This underscores the necessity and importance of studying the identifiability of models using MS. Recently, several works have studied the identifiability of the MS-based measurement models and developed recovery algorithms with guarantees. These studies can be categorized into two main types: the first involves applying MS to signals with specific properties, and the second involves preprocessing the signal to enhance its correlation before applying MS to obtain the observations. 
For bandlimited finite-energy signals, the US algorithm (USAlg) was first proposed, providing a recovery guarantee when the oversampling rate is above $2\pi e$ times the Nyquist rate \cite{bhandari2017unlimited, bhandari2020unlimited}. Additionally, \cite{bhandari2019identifiability} demonstrates that when the sampling rate exceeds the Nyquist rate, there is a one-to-one mapping between modulo samples and the original bandlimited finite-energy signal. This conclusion is also proven in \cite{romanov2019above} using an alternative method. Furthermore, \cite{romanov2019above} proposes a recovery algorithm based on linear prediction (LP) for bandlimited finite-energy signals under MS, providing a recovery guarantee when the oversampling rate is above the Nyquist rate, given the knowledge of the autocorrelation function of the signal. For sparse signals, the identifiability of the modulo compressed sensing (modulo-CS) model is studied, where a linear transform is applied first, considering both underdetermined and overdetermined setups, followed by a modulo operator \cite{prasanna2020identifiability}. The linear transform serves as a preprocessing step for the sparse signal before applying MS, offering two benefits: it enhances the correlation of the sparse signal and reduces the number of measurements needed in the underdetermined setup. In detail, the minimum number of measurements required to uniquely recover any $s$-sparse signal from modulo-CS measurements is $2s+1$. In addition, for any $s$-sparse vector, a measurement matrix with $2s+1$ rows and entries drawn independently from any continuous distribution satisfies the identifiability conditions with high probability. For periodic bandlimited (PBL) signals, the Fourier-Prony approach is proposed to handle non-ideal folding instants in practical hardware, with a recovery guarantee provided when the number of folds is known \cite{bhandari2021unlimited}. For bandpass signals, \cite{shtendel2023unlimited} demonstrates that the original signal can be recovered from the modulo samples up to an unknown constant when the sampling period is less than $\frac{1}{2\Omega_{U,L}e}$, where $\Omega_{U,L}$ is the bandwidth of the signal, achieving sub-Nyquist sampling rates, and proposes both time and frequency domain-based algorithms. In \cite{mulleti2024modulo}, the finite-rate-of-innovation (FRI) signal with $L$ pulses is first preprocessed by a sum-of-sincs (SoS) kernel, then the MS is applied. Note that the output of the SoS kernel to be sampled via MS is an $L$-order PBL signal. The authors proved that if the number of samples per period denoted as $N$ is a prime number and $N> 2L+1$, then the $L$-order PBL signal can be uniquely identified up to a constant factor from the modulo samples. Furthermore, the FRI parameters are derived from PBL samples.

 Some other works related to MS have also been studied in depth. Several algorithms on MS via rate-distortion theory have been investigated \cite{ordentlich2016integer, ordentlich2018modulo, romanov2021blind, weiss2022blind, weiss2022blind2}. Specifically, for random Gaussian distributed vectors with a known covariance matrix, the integer-forcing decoder is presented \cite{ordentlich2016integer, ordentlich2018modulo}, and a blind version is also proposed \cite{romanov2021blind}. For a stationary Gaussian stochastic process with a known auto-covariance function, the LP method is introduced and extended to a blind version \cite{ordentlich2018modulo, weiss2022blind}. For $K$ discrete-time jointly Gaussian stationary random processes with known autocorrelation functions, the LP method and integer-forcing decoder are combined to handle this case \cite{ordentlich2018modulo}, and a blind version is also introduced \cite{weiss2022blind2}. Additionally, to retrieve the original clean signal from noisy modulo samples, two-stage algorithms have been proposed, and the corresponding analysis has been studied \cite{cucuringu2020provably, fanuel2021recovering, fanuel2022denoising, tyagi2022error}. These algorithms first denoise the modulo samples before proceeding with the recovery process. Under the oversampling setup, a dynamic programming (DP) based algorithm followed by the orthogonal matching pursuit (OMP) is proposed to reconstruct the line spectral signal from noisy modulo samples, and the effectiveness of this approach is validated through real data collected by the mmWave radar \cite{zhang2024line}. The Cram\'er-Rao bounds (CRBs) are derived with known folding-count for both quantized and unquantized cases, serving as performance benchmarks \cite{cheng2023crb}. Additionally, the unlimited one-bit (UNO) sampling framework is proposed, combining the benefits of both US and dithered one-bit quantization for bandlimited signals, along with a two-stage algorithm \cite{eamaz2024uno}.

The partial discrete Fourier transform (DFT) sensing matrix, composed of randomly selected rows from a standard DFT matrix, has been widely used to reduce the number of measurements for sparse signals in the context of CS \cite{candes2006robust, ma2014turbo}. In this paper, we utilize the standard DFT matrix to preprocess general signals with amplitudes that exceed the threshold of a conventional ADC, enhancing their correlation. Subsequently, MS is applied to the preprocessed signal to ensure the amplitude of the signal before sampling remains within the ADC's threshold. This sensing model, termed the modulo-DFT sensing model, is illustrated in Fig. \ref{moduloDFTdiag}. Here, a complex-valued signal first undergoes the delay and sample-and-hold circuits followed by the DFT sensing circuit and the mod-ADC. It should be noted that applying MS on PBL signals can be viewed as a modulo-DFT sensing model with additional symmetric and conjugate constrains on the original signal when the sampling rate is greater than the Nyquist rate. For the modulo-DFT sensing model, several questions arises naturally: Is the modulo-DFT sensing model identifiable? Provided that the model is identifiable, is there any efficient recovery algorithm? These questions serve as the motivation for this work.

The contributions of this work are summarized as follows: First, we derive a necessary and sufficient condition based on the number of measurements and the indices of zero elements of original signal for the unique recovery of the original signal from the modulo measurements in the modulo-DFT sensing model. Then, we conduct a deeper analysis of three specific cases: when the number of measurements is a power of 2, a prime number, and twice a prime number. In detail, we show that if the number of measurements is a prime number and at least two elements (including the first one) of the original signal are zero, then the modulo-DFT sensing model is identifiable. Additionally, we study the identifiability of PBL signals under MS, which can be viewed as a modulo-DFT sensing model with additional symmetric and conjugate constraints on the original signal, when the sampling rate is greater than the Nyquist rate. We provide a necessary and sufficient condition with respect to the number of measurements for unique identifiability with an ambiguity in the direct current (DC) component. Moreover, we show that when the oversampling factor exceeds $3(1+1/P)$, then the PBL signal can be uniquely identified from modulo samples with an ambiguity in the DC component, where $P$ is the number of harmonics, including the fundamental component, in the positive frequency part. Given that the model is identifiable, joint recovery of the signal and folding counts can be achieved by solving linear equations involving continuous signal variables and integer folding counts. This complex problem can be equivalently transformed into integer-constrained linear equations, which reduces the computational complexity of the recovery algorithm. Finally, simulations are conducted to validate our conclusions.

\begin{figure*}[htb!]
\begin{center}
    \centering
    \includegraphics[width=0.8\textwidth]{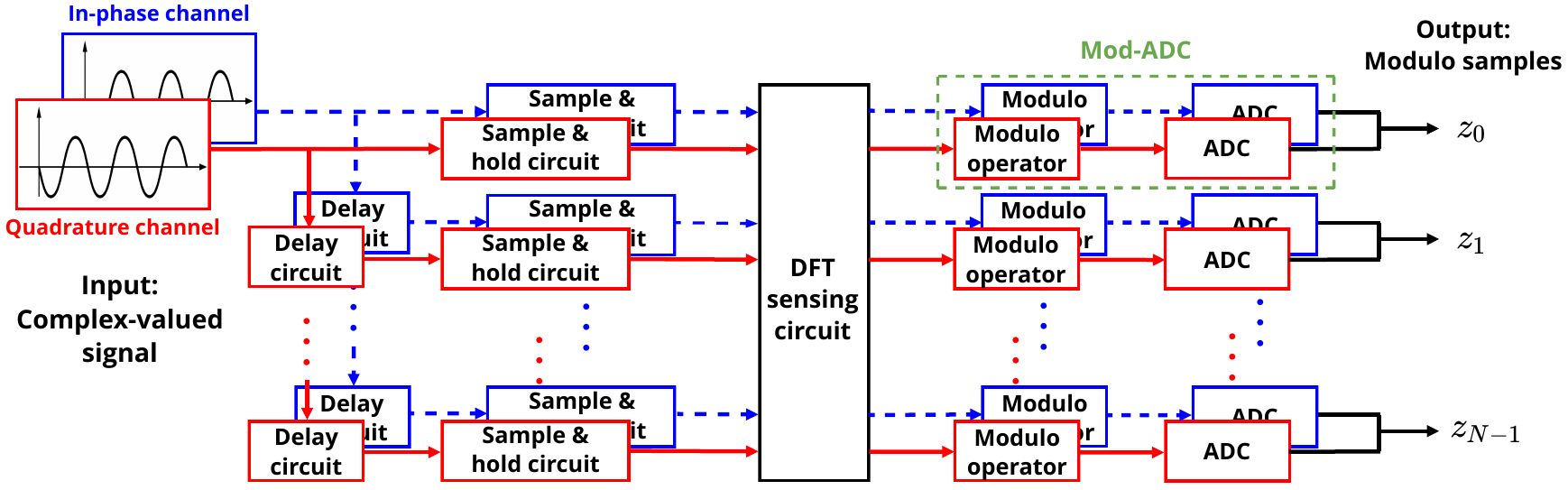}
    \caption{Schematic architecture for the modulo-DFT sensing model. A complex-valued signal undergoes the delay and sample-and-hold circuits followed by the DFT sensing circuit and the modulo circuits whose outputs are then sampled by $N$ parallel ADCs.}
    \label{moduloDFTdiag}
    \end{center}
\end{figure*}
\subsection{Notation}
Let $\mathcal{N} \triangleq \{0,1,\cdots,N-1\}$ and $\mathcal{V}$ be a subset of $\mathcal{N}$. For a complex matrix $\mathbf{F}\in\mathbb{C}^{N\times N}$, $\mathbf{F}_{\mathcal{V}}$ denotes the submatrix by choosing the columns of $\mathbf{F}$ indexed by $\mathcal{V}$\footnote{Indices for both matrices and vectors start from $0$ rather than $1$.}. Let $(\cdot)^{\rm H}$ be the Hermitian transpose operator and $\mathbf{F}^{\rm H}_{\mathcal{V}}$ denote the Hermitian transpose of $\mathbf{F}_{\mathcal{V}}$. For a complex vector $\mathbf{s}\in\mathbb{C}^{N}$, $\mathbf{s}_{\mathcal{V}}$ denotes the subvector by choosing the elements of $\mathbf{s}$ indexed by $\mathcal{V}$. We use $|\mathcal{V}|$ to denote the cardinality of $\mathcal{V}$. In addition, $\Re\{\mathbf{s}\}$ denotes the real part of $\mathbf{s}$ and $\Im\{\mathbf{s}\}$ denotes the imaginary part of $\mathbf{s}$. Let $\lfloor\cdot\rfloor$ be the floor operator and $\lceil\cdot\rceil$ be the ceiling operator. For two integers $a$ and $b$, $a \mid b$ denotes $a$ divides $b$, and $a\nmid b$ denotes $a$ does not divide $b$. Given $N$ positive integers $a_1,a_2,\cdots,a_N$, let ${\rm gcd}(a_1,a_2,\cdots,a_N)$ denote the greatest common divisor of $a_1,a_2,\cdots,a_N$. Let $f(x)$ be a polynomial and ${\rm deg}(f(x))$ denote the degree of $f(x)$. We use $\Phi_n(t)$ to denote the $n$th cyclotomic polynomial. Let $\mathbb{F}$ be a number field and $\mathbb{F}[x]\triangleq \{\sum_{i=0}^{n}a_ix^i\mid a_i\in\mathbb{F},n\in\mathbb{Z}^{+}\}$ be the ring of polynomials over $\mathbb{F}$. A Gaussian integer is a complex number such that its real and imaginary parts are both integers, and let $\mathbb{Z}[\rm {j}]\triangleq \{a+b{\rm j} \mid a, b \in \mathbb{Z}\}$ be the ring of Gaussian integers. Similarly, a Gaussian rational number is a complex number such that its real and imaginary parts are both rational numbers, and let $\mathbb{Q}[\rm {j}]\triangleq \{a+b{\rm j} \mid a, b \in \mathbb{Q}\}$ be the field of Gaussian rational numbers.

\section{Problem Setup}\label{ProbSet}
Let $\mathbf{F}\in\mathbb{C}^{N\times N}$ be the DFT matrix with the $(n_1,n_2)$th element being $\frac{1}{\sqrt{N}}{\rm e}^{-{\rm j}\frac{2\pi n_1n_2}{N}}$, where $n_1 = 0,1,\cdots,N-1$ and $n_2 = 0,1,\cdots,N-1$. The linear sensing model with the DFT sensing matrix $\mathbf{F}$ is 
\begin{align}\label{measurements}
    \mathbf{y} = \mathbf{F}\mathbf{s},
\end{align}
where $\mathbf{s}\in\mathbb{C}^{N}$ is the original complex-valued signal and $N$ is the number of measurements. Let $\mathscr{M}(\cdot): \mathbb{R}^N \mapsto [-\frac{1}{2},\frac{1}{2})^N$ be the centered modulo mapping defined as\footnote{For simplicity, we set the dynamic range of the modulo measurements as $[-\frac{1}{2},\frac{1}{2})$ in this paper. The conclusions derived in this paper are also applicable to general cases.}
\begin{align}\label{modulooperator}
    \mathscr{M}\left(\mathbf{t}\right)\triangleq \mathbf{t} - \left\lfloor\mathbf{t} + \frac{1}{2}\right\rfloor,
\end{align}
where $\mathbf{t}\in\mathbb{R}^N$ is a real vector of length $N$ and $\lfloor\cdot\rfloor$ is the element-wise floor operator. In addition, let $\mathscr{C}(\cdot)$ denote the complex centered modulo mapping, which applies the modulo operator $\mathscr{M}(\cdot)$ to both the real and imaginary parts of a complex vector, i.e., 
\begin{align}\label{commoduloope}
    \mathscr{C}(\mathbf{u})\triangleq\mathscr{M}\left(\Re{\{\mathbf{u}\}}\right) + {\rm j}\mathscr{M}\left(\Im{\{\mathbf{u}\}}\right), 
\end{align}
where $\mathbf{u}\in\mathbb{C}^N$ is a complex vector, and $\Re\{\mathbf{u}\}$ and $\Im\{\mathbf{u}\}$ return the real part and imaginary part of $\mathbf{u}$ respectively. The modulo-DFT sensing model, where the original signal is preprocessed by the DFT sensing matrix (\ref{measurements}) followed by the MS (\ref{commoduloope}), is described as 
\begin{align}\label{modulomeasurements}
    \mathbf{z} = \mathscr{C}(\mathbf{y}) = \mathscr{C}(\mathbf{F}\mathbf{s}).
\end{align}
Let $\mathbb{Z}[{\rm j}] \triangleq \{a+b{\rm j}\mid a,b\in\mathbb{Z}\}$ be the set of all Gaussian integers. It can be obtained that the modulo measurements $\mathbf{z}$ can be represented as a sum of unfolded measurements $\mathbf{y}$ and a Gaussian integer vector $\boldsymbol{\epsilon}\in\mathbb{Z}[{\rm j}]^{N}$ \cite{zhang2024line}, i.e.,
\begin{align}\label{decomsimple}
    \mathbf{z} = \mathbf{y} + \boldsymbol{\epsilon} = \mathbf{F}\mathbf{s} + \boldsymbol{\epsilon}.
\end{align}

We first demonstrate that if there is no prior information about the original signal $\mathbf{s}$, i.e., $\mathbf{s}$ can be any vector in $\mathbb{C}^N$, $\mathbf{s}$ cannot be uniquely identified from the modulo measurements $\mathbf{z}$ in the modulo-DFT sensing model (\ref{decomsimple}). In the following, we show that there exists a signal in $\mathbb{C}^N$ distinct from $\mathbf{s}$ that generates the same modulo measurements as $\mathbf{s}$. Let $\boldsymbol{\epsilon}^{\prime} \in \mathbb{Z}[{\rm j}]^{N}$ be another Gaussian integer vector distinct from $\boldsymbol{\epsilon}$ (\ref{decomsimple}), and define $\mathbf{s}^{\prime} = \mathbf{F}^{\rm H}(\mathbf{z} - \boldsymbol{\epsilon}^{\prime})$, where $(\cdot)^{\rm H}$ is the
Hermitian transpose operator. It can be verified that $\mathbf{s}^{\prime}$ produces the same modulo measurements as $\mathbf{s}$ in the modulo-DFT sensing model (\ref{modulomeasurements}). Furthermore, we have $\mathbf{s} = \mathbf{F}^{\rm H}(\mathbf{z} - \boldsymbol{\epsilon})$ according to (\ref{decomsimple}). Since $\mathbf{F}^{\rm H}$ is invertible and $\boldsymbol{\epsilon}^{\prime}$ differs from $\boldsymbol{\epsilon}$, $\mathbf{s}^{\prime}$ differs from $\mathbf{s}$. Indeed, due to the existence of infinitely many Gaussian integer vectors distinct from $\boldsymbol{\epsilon}$, for any signal $\mathbf{s}$, there exist infinitely many signals in $\mathbb{C}^{N}$ distinct from $\mathbf{s}$ that yield the same modulo measurements as $\mathbf{s}$. 

In this paper, we study the identifiability of the modulo-DFT sensing model (\ref{decomsimple}) under the condition that $\mathbf{s}$ belongs to the set $\mathbb{C}^N_{\mathcal{V}}$ which is a subspace of $\mathbb{C}^N$. Here, $\mathcal{V}$ is a subset of $\mathcal{N} \triangleq \{0, 1, \ldots, N-1\}$, and $\mathbb{C}^N_{\mathcal{V}}$ denotes the set of all complex vectors of length $N$ with elements indexed by $\mathcal{V}$ equal to zero. Since the elements of $\mathbf{s}$ indexed by $\mathcal{V}$ are $0$, $\mathbf{F}\mathbf{s}$ is equal to $\mathbf{F}_{\bar{\mathcal{V}}}\mathbf{s}_{\bar{\mathcal{V}}}$, where $\bar{\mathcal{V}}\triangleq \mathcal{N}\backslash \mathcal{V}$ is the complement of the set $\mathcal{V}$ in $\mathcal{N}$, $\mathbf{F}_{\bar{\mathcal{V}}}$ denotes the submatrix by choosing the columns of $\mathbf{F}$ indexed by $\bar{\mathcal{V}}$, and $\mathbf{s}_{\bar{\mathcal{V}}}$ denotes the subvector by choosing the elements of $\mathbf{s}$ indexed by $\bar{\mathcal{V}}$. Thus, when $\mathbf{s} \in \mathbb{C}^N_{\mathcal{V}}$, the measurement model (\ref{decomsimple}) can be written as 
\begin{align}\label{moduloequation}
    \mathbf{z}=\mathbf{y}+\boldsymbol{\epsilon}=\mathbf{F}_{\bar{\mathcal{V}}}\mathbf{s}_{\bar{\mathcal{V}}}+\boldsymbol{\epsilon}.
\end{align}

In Sec. \ref{SecIdenti}, a necessary and sufficient condition regarding the number of measurements $N$ and the set $\mathcal{V}$, which records the indices of the zero elements of $\mathbf{s}$, for uniquely recovering the original signal $\mathbf{s}$ from modulo measurements $\mathbf{z}$ in the modulo-DFT sensing model is introduced. Here, we first introduce Proposition \ref{EquiCondi} which shows that the modulo-DFT sensing model (\ref{decomsimple}) is identifiable with the constraint $\mathbf{s}\in\mathbb{C}^N_{\mathcal V}$ can be transformed into a proposition related to polynomials with Gaussian integer coefficients. In Sec. \ref{proofTheo1}, Proposition \ref{EquiCondi} is used to prove the final theorem.
\begin{proposition}\label{EquiCondi}
    The modulo-DFT sensing model (\ref{decomsimple}) is identifiable with the constraint $\mathbf{s}\in\mathbb{C}^N_{\mathcal V}$ if and only if there does not exist a nonzero Gaussian integer polynomial $f(x)$ of maximum degree $N-1$ such that $f({\rm e}^{{\rm j}\frac{2\pi n}{N}}) = 0$ for all $n\in\mathcal{V}$.
\end{proposition}
\begin{proof}
    We first prove the necessary part by proving the contrapositive, i.e., if there exists a nonzero Gaussian integer polynomial $f(x)$ of maximum degree $N-1$ such that $f({\rm e}^{{\rm j}\frac{2\pi n}{N}}) = 0$ for all $n\in\mathcal{V}$, we have that the modulo-DFT sensing model is not identifiable with the constraint $\mathbf{s}\in\mathbb{C}^N_{\mathcal V}$. According to the definition of the DFT sensing matrix $\mathbf{F}$, the existence of a nonzero Gaussian integer polynomial $f(x)$ of maximum degree $N-1$ such that $f({\rm e}^{{\rm j}\frac{2\pi n}{N}}) = 0$ for all $n\in\mathcal{V}$ is equivalent to the existence of a nonzero Gaussian integer vector $\bar{\boldsymbol{\epsilon}}$ such that $\mathbf{F}_{\mathcal{V}}^{\rm H}\bar{\boldsymbol{\epsilon}} = \mathbf{0}$, where $\mathbf{F}_{\mathcal{V}}^{\rm H}$ is the
    Hermitian transpose of $\mathbf{F}_{\mathcal{V}}$. Let $\mathbf{s'} = \mathbf{s} - \mathbf{F}^{\rm H}\bar{\boldsymbol{\epsilon}}$ which is distinct from $\mathbf{s}$ since $\bar{\boldsymbol{\epsilon}}$ is a nonzero Gaussian integer vector. In addition, since $\mathbb{C}^N_{\mathcal{V}}$ is closed under addition and $\mathbf{F}_{\mathcal{V}}^{\rm H}\bar{\boldsymbol{\epsilon}} = \mathbf{0}$, it follows that $\mathbf{s'}\in\mathbb{C}^N_{\mathcal{V}}$. Furthermore, the modulo measurements generated by $\mathbf{s'}$ are 
    \begin{align}
        \mathscr{C}(\mathbf{F}\mathbf{s}') = \mathscr{C}(\mathbf{F}(\mathbf{s}-\mathbf{F}^{\rm H}\bar{\boldsymbol{\epsilon}})) = \mathscr{C}(\mathbf{F}\mathbf{s} - \bar{\boldsymbol{\epsilon}}) 
    \end{align}
    which equal to those generated by $\mathbf{s}$, i.e., $\mathscr{C}(\mathbf{F}\mathbf{s})$, since $\mathscr{C}(\cdot)$ is ambiguous for inputs that differ by a Gaussian integer vector. Thus, the modulo-DFT sensing model is not identifiable.

    Now, we prove the sufficient part by proving the contrapositive, i.e., if the modulo-DFT sensing model is not identifiable with the constraint $\mathbf{s}\in\mathbb{C}^N_{\mathcal V}$, then there exists a nonzero Gaussian integer polynomial $f(x)$ of maximum degree $N-1$ such that $f({\rm e}^{{\rm j}\frac{2\pi n}{N}}) = 0$ for all $n\in\mathcal{V}$. Let $\mathbf{s}'$ be a vector in $\mathbb{C}^N_{\mathcal V}$ distinct from $\mathbf{s}$ that generates the same modulo measurements as $\mathbf{s}$, i.e., $\mathscr{C}(\mathbf{F}\mathbf{s}) = \mathscr{C}(\mathbf{F}\mathbf{s}')$. Based the definition of the modulo operator $\mathscr{C}(\cdot)$ (\ref{commoduloope}) and the invertibility of $\mathbf{F}$, we have 
    \begin{align}\label{equationdiff}
        \mathbf{F}\mathbf{s} - \mathbf{F}\mathbf{s}' = \bar{\boldsymbol{\epsilon}},
    \end{align}
    where $\bar{\boldsymbol{\epsilon}}$ is a nonzero Gaussian integer vector. Furthermore, multiplying equation (\ref{equationdiff}) by $\mathbf{F}^{\rm H}$ yields $\mathbf{s} - \mathbf{s}' =  \mathbf{F}^{\rm H}\bar{\boldsymbol{\epsilon}}$. Since both $\mathbf{s}$ and $\mathbf{s}'$ are in $\mathbb{C}^N_{\mathcal V}$, we have $\mathbf{F}_{\mathcal{V}}^{\rm H}\bar{\boldsymbol{\epsilon}} = \mathbf{0}$, which is equivalent to the existence of a nonzero Gaussian integer polynomial $f(x)$ of maximum degree $N-1$ such that $f({\rm e}^{{\rm j}\frac{2\pi n}{N}}) = 0$ for all $n\in\mathcal{V}$ as previously illustrated.
\end{proof}

\section{Preliminaries}\label{SecPreli}
The identifiability of the modulo-DFT sensing model (\ref{decomsimple}) is closely related to Gaussian integer polynomials as indicated by Proposition \ref{EquiCondi}. This section introduces some properties and propositions of polynomials in number theory, which is beneficial in establishing the necessary and sufficient condition for uniquely identifying the modulo-DFT sensing model and the PBL signal via MS in the ensuing sections. Let $p$ be a prime number and $f(x) = x^{p-1} + x^{p-2} + \cdots + 1$, Property \ref{LemGauRat} shows that $f(x)$ is irreducible over the field of Gaussian rationals $\mathbb{Q}[{\rm j}]$, where $\mathbb{Q}[\rm {j}] \triangleq \{a+b{\rm j} \mid a, b \in \mathbb{Q}\}$. In addition, Property \ref{LemPowGau} shows that $f(x) = x^{2^m} + {\rm j}$ or $f(x) = x^{2^m} - {\rm j}$ is also irreducible over the field of Gaussian rationals $\mathbb{Q}[{\rm j}]$, where $m$ is a nonnegative integer.
\begin{property}\label{LemGauRat}
    Let $p$ be a prime number. The polynomial $f(x) = x^{p-1}+x^{p-2}+\cdots+1$ is irreducible over the field of Gaussian rationals $\mathbb{Q}[{\rm j}]$. In other words, $f(x)$ can not be represented as the factor of two nonconstant polynomials with Gaussian rational coefficients.
\end{property}
\begin{proof}
    The proof is postponed to Appendix \ref{AppenA}.
\end{proof}
\begin{property}\label{LemPowGau}
    Let $m$ be a nonnegative integer. The polynomial $f(x) = x^{2^m} - {\rm j}$ or $f(x) = x^{2^m} + {\rm j}$ is irreducible over the field of Gaussian rationals $\mathbb{Q}[{\rm j}]$.
\end{property}
\begin{proof}
    The proof is postponed to Appendix \ref{AppenB}.
\end{proof}

Let $f(x)$ be a polynomial over a number field $\mathbb{F}$. Proposition \ref{UniFacPoly} shows the uniqueness of factorization of $f(x)$ in $\mathbb{F}[x]$, where $\mathbb{F}[x] \triangleq \{\sum_{i=0}^{n}a_ix^i\mid a_i\in\mathbb{F},n\in\mathbb{N}\}$ is the ring of polynomials over $\mathbb{F}$. In addition, Proposition \ref{polyReTh} shows the polynomial remainder theorem. In detail, let $f(x)$ and $g(x)$ be two nonzero polynomials in $\mathbb{F}[x]$, then there exist unique polynomials $q(x)$ and $r(x)$ such that $f(x) = g(x)q(x)+ r(x)$ with ${\rm deg}(r(x))<{\rm deg}(g(x))$, where ${\rm deg}(\cdot)$ returns the degree of a polynomial. Let $\xi\in\mathbb{C}$ be any root of $f(x)$. Proposition \ref{CycPriP} shows that if $f(x)$ is irreducible over $\mathbb{F}$, there is no polynomial in $\mathbb{F}[x]$ of degree less than ${\rm deg}(f(x))$ having a root at $\xi$.

\begin{proposition}\cite{gardner}\label{UniFacPoly}
   Every nonconstant polynomial $f(x) \in \mathbb{F}[x]$ can be factored in $\mathbb{F}[x]$ into a product of irreducible polynomials, the irreducible polynomials being unique except for order and for unit (that is, nonzero constant) factors in $\mathbb{F}$.
\end{proposition}

\begin{proposition}\cite[Theorem $16.2$]{gallian2021contemporary}(Polynomial remainder theorem)\label{polyReTh}
   Let $f(x)$ and $g(x)$ be two nonzero polynomials in $\mathbb{F}[x]$. Then there exist unique polynomials $q(x)$ and $r(x)$ in $\mathbb{F}[x]$ such that $f(x) = g(x)q(x) + r(x)$ with $\operatorname{deg}(r(x))<\operatorname{deg} (g(x))$. 
\end{proposition}

\begin{proposition}\label{CycPriP}
 Let $f(x)$ be a polynomial over the number field $\mathbb{F}$ and $\xi$ be any root of $f(x)$. If $f(x)$ is irreducible over $\mathbb{F}$, then there does not exist a polynomial in $\mathbb{F}[x]$ of degree less than ${\rm deg}(f(x))$ having a root at $\xi$.
\end{proposition}
\begin{proof}
The proof is postponed to Appendix \ref{AppenC}.
\end{proof}

\section{Identifiability of Modulo-DFT Sensing Model}\label{SecIdenti}
In this section, we provide a necessary and sufficient condition concerning the number of measurements $N$ and the set $\mathcal{V}$, which records the indices of the zero elements of the original signal $\mathbf{s}$, for uniquely identifying $\mathbf{s}$ from the modulo measurements $\mathbf{z}$ in the modulo-DFT sensing model (\ref{decomsimple}). Before presenting the final conclusion, we introduce several definitions below.

We first consider the factorization of the polynomial $x^N-1$, which is crucial for discussing the identifiability of the modulo-DFT sensing model (\ref{decomsimple}). Let $C(x)\triangleq x^N - 1$. According to Proposition \ref{UniFacPoly}, $C(x)$ can be factored into a unique product of irreducible monic polynomials in $\mathbb{Q}[{\rm j}][x]$ without considering the order, where a monic polynomial is a non-zero univariate polynomial in which the leading coefficient is equal to $1$ and $\mathbb{Q}[{\rm j}][x] \triangleq \{\sum_{i=0}^{n}a_ix^i\mid a_i\in\mathbb{Q}[{\rm j}],n\in\mathbb{N}\}$ is the ring of polynomials over $\mathbb{Q}[{\rm j}]$. In detail, we have 
\begin{align}\label{factorCx}
    C(x) =x^N-1= c_1(x) \cdots c_K(x),
\end{align}
where $c_1(x),\cdots,c_K(x)$ are all monic polynomials in $\mathbb{Q}[{\rm j}][x]$ and are all irreducible over $\mathbb{Q}[{\rm j}]$. It is well known that $C(x)$ has $N$ distinct roots $\{{\rm e}^{{\rm j} \frac{2\pi n}{N}}\}_{n=0}^{N-1}$, called $N$ $N$th roots of unity. Let $\mathcal{R}$ be the set of all $N$th roots of unity, that is,
\begin{align}\label{unitroots}
    \mathcal{R}\triangleq\left\{1, {\rm e}^{{\rm j} \frac{2\pi\times 1}{N}}, \cdots, {\rm e}^{{\rm j}\frac{2\pi \times (N-1)}{N}}\right\}.
\end{align}
In addition, let $\mathcal{C}_k$ be the set of roots of $c_k(x)$, i.e.,
\begin{align}\label{unitrootsPar}
    \mathcal{C}_k \triangleq \{x\mid c_k(x)=0\}, k=1,\cdots,K.
\end{align}
Given that $C(x) = x^N - 1$ possesses $N$ distinct roots as indicated in (\ref{unitroots}), $C(x)$ does not have any repeated roots. Consequently, the roots of $c_1(x), \cdots, c_K(x)$ constitute a partition of the roots of $C(x)$ since $C(x) = c_1(x) \cdots c_K(x)$ (\ref{factorCx}), implying that $\mathcal{C}_1, \cdots, \mathcal{C}_K$ form a partition of $\mathcal{R}$. In detail, $\mathcal{R} = \bigcup_{k=1}^K\mathcal{C}_k$ and $\mathcal{C}_1,\cdots,\mathcal{C}_K$ are pairwise disjoint, i.e., $\mathcal{C}_i\cap\mathcal{C}_j=\emptyset$ for any $1\leq i<j\leq K$. Note that each element in $\mathcal{R}$ is the $n$th power of ${\rm e}^{{\rm j}\frac{2\pi}{N}}$, where $n$ is an integer in the set $\mathcal{N}=\{0,1,\cdots,N-1\}$. Let $\mathcal{N}_k$ be the set that records the power of ${\rm e}^{{\rm j}\frac{2\pi}{N}}$ in $\mathcal{C}_k$, which is
\begin{align}\label{defni}
    \mathcal{N}_k \triangleq \left\{n \mid {\rm e}^{{\rm j}\frac{2\pi n}{N}}\in\mathcal{C}_k, n\in\mathcal{N}\right\},k=1,\cdots,K.
\end{align}
Since $\mathcal{C}_1,\cdots,\mathcal{C}_K$ form a partition of $\mathcal{R}$, $\mathcal{N}_1,\cdots,\mathcal{N}_K$ also form a partition of $\mathcal{N}$. Here, we provide a specific example to aid understanding. For $N=4$, we have $C(x) = x^4-1 = (x-1)(x+1)(x-{\rm j})(x+{\rm j})$. Hence, the number of factors is $K=4$ and the factors are $c_1 = x-1$, $c_2 = x+1$, $c_3=x-{\rm j}$, and $c_4=x+{\rm j}$. Furthermore, the sets of roots for different factors are $\mathcal{C}_1 = \{1\}$, $\mathcal{C}_2 = \{-1\}$, $\mathcal{C}_3 = \{{\rm j}\}$, and $\mathcal{C}_4 = \{-{\rm j}\}$, and the corresponding sets that show the power of ${\rm e}^{{\rm j}\frac{2\pi}{4}}$ in $\mathcal{C}_k,k=1,2,3,4$ are $\mathcal{N}_1 = \{0\}$, $\mathcal{N}_2 = \{2\}$, $\mathcal{N}_3 = \{1\}$, and $\mathcal{N}_4 = \{3\}$.

Based on the definitions above, a necessary and sufficient condition for uniquely recovering $\mathbf{s}$ from $\mathbf{z}$ is proposed in Theorem \ref{TheoComplex}, and the proof of Theorem \ref{TheoComplex} is provided at the end of this section, i.e., Sec. \ref{proofTheo1}.
\begin{theorem}\label{TheoComplex}(A necessary and sufficient condition)
    In the modulo-DFT sensing model (\ref{decomsimple}), the signal $\mathbf{s}\in\mathbb{C}^N_{\mathcal{V}}$ can be uniquely recovered from the modulo measurements $\mathbf{z}$ if and only if $\forall k\in\{1,\cdots,K\}$, $\mathcal{N}_k\cap\mathcal{V}\neq \emptyset$.
\end{theorem}

Note that for any $N$, $C(x)$ has a factor $x-1$ since $C(x) = x^N-1 = (x-1)(x^{N-1}+\cdots+1)$. Without loss of generality, let $c_1(x) =x-1$, we have $\mathcal{C}_1 = \{1={\rm e}^{{\rm j}\frac{2\pi\times 0}{N}}\}$ and $\mathcal{N}_1 = \{0\}$ according to the definitions (\ref{unitrootsPar}) and (\ref{defni}). As the necessary and sufficient condition in Theorem \ref{TheoComplex} is $\forall k\in\{1,\cdots,K\}$, $\mathcal{N}_k\cap\mathcal{V}\neq \emptyset$, we have $\mathcal{N}_1\cap\mathcal{V}=\{0\}\cap\mathcal{V} \neq \emptyset$. Thus, for $\mathbf{s}$ to be uniquely recovered from $\mathbf{z}$ for any $N$, it is necessary for $\mathcal{V}$ to contain $0$. This conclusion can also be obtained by directly observing the modulo-DFT sensing model (\ref{modulomeasurements}). Since the first column of $\mathbf{F}$ is $\frac{1}{\sqrt{N}}[1, 1, \cdots, 1]^{\rm T}$, which is a constant vector, and the output of the modulo operator (\ref{commoduloope}) remains unchanged when the input has a constant integer offset, $\mathbf{s}$ and $\mathbf{s} + \sqrt{N}\mathbf{e}_0$ generate the same modulo measurements. Here, $\mathbf{e}_0$ denotes the one-hot vector with the $0$th element being $1$ and the others being $0$. Therefore, the $0$th element of $\mathbf{s}$, i.e., $s_0$, is unidentifiable. However, upon additional examination, we observe that the index $0$ solely impacts the identiﬁability of $s_0$. This implies that ${s}_1,{s}_2,\cdots,{s}_{N-1}$ can still be uniquely recovered from $\mathbf{z}$ if and only if $\forall k\in\{2,\cdots,K\}$, $\mathcal{N}_k\cap\mathcal{V}\neq \emptyset$. This conclusion is outlined in Corollary \ref{CoroTheorem}.
\begin{corollary}\label{CoroTheorem}
        Let $\mathbf{z}$ be the modulo measurements of signal $\mathbf{s}\in\mathbb{C}^N_{\mathcal{V}}$ in the modulo-DFT sensing model (\ref{decomsimple}). $\mathbf{s}_{1:N-1}$ can be uniquely recovered from $\mathbf{z}$ if and only if $\forall k\in\{2,\cdots,K\}$, $\mathcal{N}_k\cap\mathcal{V}\neq \emptyset$.
\end{corollary}
\begin{proof}
When $0\in\mathcal{V}$, this can be easily proved according to Theorem \ref{TheoComplex}. When $0 \notin \mathcal{V}$, similar to Proposition \ref{EquiCondi}, an equivalent condition for unique identifiability can be derived as shown in Proposition \ref{EquiCondi2}. Furthermore, using a proof strategy similar to that of Theorem \ref{TheoComplex}, Corollary \ref{CoroTheorem} can be proven.
\end{proof}
\begin{proposition}\label{EquiCondi2}
Let $\mathbf{z}$ be the modulo measurements of signal $\mathbf{s}\in\mathbb{C}^N_{\mathcal{V}}$ in the modulo-DFT sensing model (\ref{decomsimple}). $\mathbf{s}_{1:N-1}$ can be uniquely recovered from $\mathbf{z}$ if and only if there does not exist a Gaussian integer polynomial $f(x)$ of maximum degree $N-1$ not in $\{A\sum_{n=0}^{N-1}x^n\mid A\in\mathbb{Z}[{\rm j}]\}$ satisfying $f({\rm e}^{{\rm j}\frac{2\pi n}{N}}) = 0$ for all $n\in\mathcal{V}$. 
\end{proposition}

In Theorem \ref{TheoComplex}, the necessary and sufficient condition for unique identifiability is that $\forall k\in\{1,\cdots,K\}$, $\mathcal{N}_k\cap\mathcal{V}\neq \emptyset$. Note that sets $\mathcal{N}_k,k=1,\cdots,K$ (\ref{defni}) are determined by the factorization of $C(x)$ over $\mathbb{Q}[{\rm j}]$, which depends on the value of $N$. Therefore, the necessary and sufficient condition is determined by the number of measurements $N$ and the set of indices of zero elements in $\mathbf{s}$, that is, $\mathcal{V}$. In fact, studying the factorization of $C(x)$ over $\mathbb{Q}[{\rm j}]$ in general is rather challenging. In the following, we conduct a deeper analysis of several specific cases.

\textbf{{Case \RomanNumeralCaps{1}:}} Let $M$ be a positive integer. We first consider the case when $N$ is the $M$th power of $2$, that is, $N = 2^M$. In the following, we first consider the factorization of $C(x)$ and then obtain the specific definition of $\mathcal{N}_k$. It can be verified that $C(x)$ is factored as 
\begin{align}\label{factorCx2N}
    C(x) &= x^N-1 = x^{2^M}-1 \notag\\
    &= (x-1)(x+1)\prod_{m=0}^{M-2}\left[(x^{2^m}-{\rm j})(x^{2^m}+{\rm j})\right],
\end{align}
over $\mathbb{Q}[{\rm j}]$ since $x^{2^m}-{\rm j}$ and $x^{2^m}+{\rm j}$ are all irreducible over $\mathbb{Q}[{\rm j}]$ for $m=0,1,\cdots,M-2$ according to Property \ref{LemPowGau}. Thus, the number of factors for $C(x)$ is $K=2M$, with the factors specified as $c_1(x) = x-1$, $c_2(x) = x+1$, $c_{2m+3}(x) = x^{2^m}-{\rm j}$, and $c_{2m+4}(x) = x^{2^m}+{\rm j}$ for $m=0,1,\cdots,M-2$. Next, we consider the set of roots of $c_{k}(x)$, that is, $\mathcal{C}_k$ (\ref{unitrootsPar}), for $k=1,2,\cdots,2M$. When $k=1$ and $k=2$, we can easily obtain $\mathcal{C}_1 = \{1\}$ and $\mathcal{C}_2 = \{-1\}$. In the following, we derive the sets $\mathcal{C}_{2m+3} \text{ and }\mathcal{C}_{2m+4}, m=0,1,\cdots,M-2$. First, we have $\mathcal{C}_{2m+3} = \{x\mid c_{2m+3}(x)=0\} = \{x\mid x^{2^m}={\rm j}\}$. Note that if $r$ is a root of $x^{2^m}=1$, then $r{\rm e}^{{\rm j}\frac{\pi}{2^{m+1}}}$ is a root of $x^{2^m}={\rm j}$. As the set of roots of $x^{2^m}=1$ is $\{{\rm e}^{{\rm j}\frac{2\pi i}{2^m}}\}_{i=0}^{2^m-1}$, we have 
\begin{align}\label{equC2m1}
    \mathcal{C}_{2m+3} = \left\{{\rm e}^{{\rm j}\left(\frac{2\pi i}{2^m}+\frac{\pi}{2^{m+1}}\right)}\right\}_{i=0}^{2^m-1}, 
\end{align}
where $m=0,1,\cdots,M-2$. Note that the elements of $\mathcal{C}_{2m+3}$ are evenly distributed on the unit circle in the complex plane.
Similarly, we have 
\begin{align}\label{equC2m2}
    \mathcal{C}_{2m+4} = \left\{{\rm e}^{{\rm j}\left(\frac{2\pi i}{2^m}+\frac{3\pi}{2^{m+1}}\right)}\right\}_{i=0}^{2^m-1},
\end{align}
where $m=0,1,\cdots,M-2$. Now we consider the sets $\mathcal{N}_k$ (\ref{defni}), which records the power of ${\rm e}^{{\rm j}\frac{2\pi}{2^M}}$ in $\mathcal{C}_k$, for $k=1,2,\cdots,2M$. First, it is obvious that $\mathcal{N}_1 = \{0\}$ and $\mathcal{N}_2=\{2^{M-1}\}$ as $\mathcal{C}_1 = \{1\}$ and $\mathcal{C}_2 = \{-1\}$. In addition, according to $\mathcal{C}_{2m+3}$ (\ref{equC2m1}), we have 
\begin{align}\label{defN2m3}
    \mathcal{N}_{2m+3} = \left\{i 2^{M-m} + 2^{M-m-2}\right\}_{i=0}^{2^m-1}, 
\end{align}
where $m=0,1,\cdots,M-2.$ Given $m$, $\mathcal{N}_{2m+3}$ is the set that contains $2^m$ elements evenly spaced within $\mathcal{N}=\{0,1,\cdots,2^M-1\}$, specifically starting from $2^{M-m-2}$ with distance $2^{M-m}$.
Similarly, according to $\mathcal{C}_{2m+4}$ (\ref{equC2m2}), we have 
\begin{align}\label{defN2m4}
    \mathcal{N}_{2m+4} = \left\{i 2^{M-m} + 3\times 2^{M-m-2}\right\}_{i=0}^{2^m-1},
\end{align}
where $m=0,1,\cdots,M-2$. In summary, we have $\mathcal{N}_{1} = \{0\}$, $\mathcal{N}_{2} = \{\frac{N}{2} = 2^{M-1}\}$, $\mathcal{N}_{3} = \left\{2^{M-2}\right\}$, $\mathcal{N}_{4} = \left\{3\times 2^{M-2}\right\}$ ,$\cdots$, $\mathcal{N}_{2m+3} = \left\{i 2^{M-m} + 2^{M-m-2}\right\}_{i=0}^{2^m-1}$, $\mathcal{N}_{2m+4} = \left\{i 2^{M-m} + 3\times 2^{M-m-2}\right\}_{i=0}^{2^m-1}$, $\cdots$, $\mathcal{N}_{2M-1} = \left\{i4 + 1\right\}_{i=0}^{2^{M-2}-1}$, $\mathcal{N}_{2M} = \left\{4i + 3\right\}_{i=0}^{2^{M-2}-1}$. The conclusion is summarized in Corollary \ref{Case12M}.
\begin{corollary}\label{Case12M}
    When $N=2^M$ where $M\in\mathbb{Z}^{+}$ is a positive integer, $\mathbf{s}\in\mathbb{C}^N_{\mathcal{V}}$ can be uniquely recovered from $\mathbf{z}$ if and only if $\forall k\in\{1,\cdots,2M\}$, $\mathcal{N}_k\cap\mathcal{V}\neq \emptyset$, where $\mathcal{N}_{1} = \{0\}$, $\mathcal{N}_{2} = \{2^{M-1}\}$, and $\mathcal{N}_{2m+3}$ and $\mathcal{N}_{2m+4}$ are defined in (\ref{defN2m3}) and (\ref{defN2m4}) for $m=0,1,\cdots,M-2$. 
\end{corollary}
According to Corollary \ref{Case12M}, the minimum number of zero elements in $\mathbf{s}$ is $2M$ to ensure that $\mathbf{s}$ can be uniquely recovered from modulo observations $\mathbf{z}$, and the positions of these zero elements are specifically required and are scattered throughout the entire set $\mathcal{N} = \{0,1,\cdots,2^M-1\}$. For a better understanding of the necessary and sufficient condition in Corollary \ref{Case12M}, in the following, we provide an example where $M=4$, that is, $N=16$. In this case, we have $\mathcal{N}_1 = \{0\}, \mathcal{N}_2 = \{8\},\mathcal{N}_3 = \{4\}, \mathcal{N}_4 = \{12\}, \mathcal{N}_5 = \{2,10\}, \mathcal{N}_6 = \{6,14\}, \mathcal{N}_7 = \{1,5,9,13\}, \mathcal{N}_8 = \{3,7,11,15\}$. It is required that $\mathcal{V}$ contains at least one element of each $\mathcal{N}_k$ for $k=1,\cdots,K$ to ensure $\mathbf{s}$ can be uniquely recovered. For example, the set $\mathcal{V} = \{0,1,2,3,4,6,8,12\}$ satisfies the requirement. In contrast, the set $\mathcal{V} = \{0,1,3,4,8,12\}$ does not satisfy the requirement. In particular, when $\mathcal{V} = \{0,1,3,4,8,12\}$, $\forall \mathbf{s}\in \mathbb{C}^N_{\mathcal{V}}$, $\mathbf{s}' = \mathbf{s} + 4\mathbf{e}_{2} + 4\mathbf{e}_{10}$, which also belongs to $\mathbb{C}^N_{\mathcal{V}}$, generates the same modulo measurements as $\mathbf{s}$ does. Here, $\mathbf{e}_i$ denotes the one-hot vector with the $i$th element being $1$ and the others being $0$.

\textbf{{Case \RomanNumeralCaps{2}:}} We consider another case when $N$ is a prime number. In this case, $C(x)$ is factored as $C(x) = (x-1)(x^{N-1} + x^{N-2} + \cdots+1)$ over $\mathbb{Q}[{\rm j}]$ since $x^{N-1} + x^{N-2} + \cdots + 1$ is irreducible over the field $\mathbb{Q}[{\rm j}]$ according to Property \ref{LemGauRat}. Thus, the number of factors of $C(x)$ is $K=2$, and the factors are $c_1(x) = x-1$, and $c_2(x) = x^{N-1} + x^{N-2} + \cdots+1$. In addition, it can be easily obtained that $\mathcal{N}_1 = \{0\}$ and $\mathcal{N}_2 = \{1,\cdots,N-1\}$ according to the definition of $\mathcal{N}_k$ (\ref{defni}). Therefore, according to Theorem \ref{TheoComplex}, $\mathbf{s}\in\mathbb{C}^N_{\mathcal{V}}$ can be uniquely recovered from $\mathbf{z}$ if and only if $0\in\mathcal{V}$ and $\exists n\in\{1,\cdots,N-1\}$ such that $n\in\mathcal{V}$ which is equivalent to $|\mathcal{V}|\geq 2$ and $0\in\mathcal{V}$. The conclusion is summarized in Corollary \ref{FirstCondiPri}. It is worth noting that when the number of measurements $N$ is a prime number, the condition for the modulo-DFT sensing model to be identifiable can be easily satisfied in practical scenarios by setting the first two original measurements ($s_0$ and $s_1$) to $0$, or by setting the first and last original measurements ($s_0$ and $s_{N-1}$) to $0$. In contrast, if we remove the DFT sensing matrix and directly perform the modulo operation on the original signal $\mathbf{s}$, there is no guarantee that $\mathbf{s}$ can be uniquely recovered from the modulo measurements. This demonstrates that performing DFT prepprocessing before MS benefits identifying the original signal.
\begin{corollary}\label{FirstCondiPri}
    When $N$ is a prime number, $\mathbf{s}\in\mathbb{C}^N_{\mathcal{V}}$ can be uniquely recovered from $\mathbf{z}$ if and only if $|\mathcal{V}|\geq 2$ and $0\in\mathcal{V}$.
\end{corollary}

\textbf{{Case \RomanNumeralCaps{3}:}} Finally, we consider the case that $N = 2p$ where $p$ is a prime number and $p\geq 3$.\footnote{When $p=2$, we have $N=2^2$, a case that has been examined in Case \RomanNumeralCaps{1}.} Note that $C(x)= x^N - 1 = x^{2p}-1$ equals to
\begin{align}\label{factorCx2p}
    C(x)  = (x-1)(x+1)\left(\sum_{i=0}^{p-1}x^{i}\right)\left(\sum_{i=0}^{p-1}(-x)^{i}\right),        
\end{align}
which can be verified directly. First, according to Property \ref{LemGauRat}, $\sum_{i=0}^{p-1}x^{i}$ is irreducible over $\mathbb{Q}[{\rm j}]$. Furthermore, $\sum_{j=0}^{p-1}(-x)^{i}$ is also irreducible over $\mathbb{Q}[{\rm j}]$ because if $f(x)$ is
irreducible over $\mathbb{Q}[{\rm j}]$ then so is $f(-x)$. Therefore, equation (\ref{factorCx2p}) is the unique factorization of $C(x) = x^{2p}-1$ over $\mathbb{Q}[{\rm j}]$. Thus, the number of factors of $C(x)$ is $K=4$, and the factors are $c_1(x) = x-1$, $c_2(x) = x+1$, $c_3(x) = \sum_{i=0}^{p-1}x^{i}$, and $c_4(x) = \sum_{i=0}^{p-1}(-x)^{i}$. In addition, the sets of roots of $c_{k}(x)$, $k=1,2,3,4$ are $\mathcal{C}_1 = \{1\}$, $\mathcal{C}_2 = \{-1\}$, $\mathcal{C}_3 = \{{\rm e}^{{\rm j}\frac{2\pi i}{p}}\}_{i=1}^{p-1}$ and $\mathcal{C}_4 = \{{\rm e}^{{\rm j}\frac{2\pi (2i+1)}{2p}}\}_{i=0,2i+1\neq p}^{p-1}$ with cardinalities of $1$, $1$, $p-1$, and $p-1$, respectively. Furthermore, we have $\mathcal{N}_1 = \{0\}$, $\mathcal{N}_2 = \{p\}$, $\mathcal{N}_3 = \{2,4,\cdots,2p-2\}$, and $\mathcal{N}_4 = \{1,3,\cdots,p-2,p+2,\cdots,2p-1\}$ which indicates the exponent of ${\rm e}^{{\rm j}\frac{2\pi}{2p}}$ in $\mathcal{C}_1$, $\mathcal{C}_2$, $\mathcal{C}_3$, and $\mathcal{C}_4$, respectively. Therefore, according to Theorem \ref{TheoComplex}, $\mathbf{s}\in\mathbb{C}^N_{\mathcal{V}}$ can be uniquely recovered from $\mathbf{z}$ if and only if $\{0,p, n_1,n_2\}\subseteq\mathcal{V}$ where $n_1\in\{2,4,\cdots,2p-2\}$ is an even number in $\mathcal{N} = \{0,1,\cdots,N-1\}$ other than $0$ and $n_2\in\{1,3,\cdots,p-2,p+2,\cdots,2p-1\}$ is an odd number in $\mathcal{N}$ other than $p$. The conclusion is summarized in Corollary \ref{FirstCondiPri2}.
\begin{corollary}\label{FirstCondiPri2}
    When $N=2p$ where $p\geq 3$ is a prime number, $\mathbf{s}\in\mathbb{C}^N_{\mathcal{V}}$ can be uniquely recovered from $\mathbf{z}$ if and only if $\{0,p, n_1,n_2\}\subseteq\mathcal{V}$ where $n_1\in\{2,4,\cdots,2p-2\}$ and $n_2\in\{1,3,\cdots,p-2,p+2,\cdots,2p-1\}$.
\end{corollary}
\subsection{Recovery algorithm}
In this subsection, a recovery algorithm that estimates the original signal $\mathbf{s}$ from modulo measurements $\mathbf{z}$ with known $\mathcal{V}$ is proposed. To estimate $\mathbf{s}$ from $\mathbf{z}$, according to equation (\ref{decomsimple}), this is equivalent to solving the following system of equations for $\mathbf{s}$ and $\boldsymbol{\epsilon}$:
\begin{align}\label{equationsolver1}
\mathbf{z} = \mathbf{F}\mathbf{s} + \boldsymbol{\epsilon},
\end{align}
where $\mathbf{s}$ is constrained to $\mathbb{C}_{\mathcal{V}}^{N}$ and $\boldsymbol{\epsilon}\in\mathbb{Z}[{\rm j}]^{N}$ is a Gaussian integer vector. Multiplying (\ref{equationsolver1}) by $\mathbf{F}^{\rm H}$ yields
\begin{align}\label{equationsolver2}
\mathbf{F}^{\rm H}\mathbf{z} = \mathbf{s} + \mathbf{F}^{\rm H}\boldsymbol{\epsilon},
\end{align}
which is equivalent to (\ref{equationsolver1}) due to the invertibility of  $\mathbf{F}$. Since $\mathbf{s}$ is any vector in $\mathbb{C}_{\mathcal{V}}^{N}$, given $\mathcal{V}$, solving (\ref{equationsolver2}) is equivalent to solving
\begin{align}\label{equationsolver3}
\mathbf{F}_{\mathcal{V}}^{\rm H}\mathbf{z} =\mathbf{F}_{\mathcal{V}}^{\rm H}\boldsymbol{\epsilon}.
\end{align}
Given the estimated Gaussian integer vector $\widehat{\boldsymbol{\epsilon}}$, the original signal $\mathbf{s}$ can be estimated as $\widehat{\mathbf{s}} = \mathbf{F}^{\rm H}(\mathbf{z} - \widehat{\boldsymbol{\epsilon}})$. Therefore, recovering the original signal $\mathbf{s}$ from the modulo measurements $\mathbf{z}$ with the known set $\mathcal{V}$ is equivalent to solving the integer linear equations given in (\ref{equationsolver3}). Practical solvers, such as SCIP \cite{bolusani2024scip} and Gurobi \cite{gurobi}, can be employed to solve these integer linear equations (\ref{equationsolver3}).

\subsection{Simulation}\label{Simulation1}
In this subsection, we conduct a simulation to validate our conclusions. We use the professional optimization software Gurobi to solve the integer linear equations (\ref{equationsolver3}) and then estimate the original signal $\mathbf{s}$. 

Both the real and imaginary parts of the nonzero elements of $\mathbf{s}$ are obtained from a uniform distribution over the interval $[-1, 1]$, and the dynamic range of the ADC is $[-0.5, 0.5]$. For the recovery algorithm, both the real and imaginary parts of each element of $\boldsymbol{\epsilon}$ are constrained to $\{-1,0,1\}$. We evaluate the probabilities of successful recovery of the original signal $\mathbf{s}$ using the proposed recovery algorithm over $300$ Monte Carlo trials for $N = 5, 6, 7, 8, 10, 11, 14, 16$. Note that $N = 8, 16$ are powers of $2$, corresponding to Case \RomanNumeralCaps{1}; $N = 5, 7, 11$ are prime numbers, corresponding to Case \RomanNumeralCaps{2}; and $N = 6, 10, 14$ are twice prime numbers, corresponding to Case \RomanNumeralCaps{3}. We establish two different scenarios, named Scenario 1 and Scenario 2, to generate $\mathcal{V}$, the set that records the indices of the zero elements of the original signal $\mathbf{s}$. For Scenario $1$, $\mathcal{V}$ satisfies the necessary and sufficient conditions outlined in Theorem \ref{TheoComplex}. Specifically, when $N = 8, 16$, $\mathcal{V}$ satisfies the necessary and sufficient condition in Corollary \ref{Case12M}. When $N = 5, 7, 11$, $\mathcal{V}$ satisfies the necessary and sufficient condition in Corollary \ref{FirstCondiPri}. When $N = 6, 10, 14$, $\mathcal{V}$ satisfies the necessary and sufficient condition in Corollary \ref{FirstCondiPri2}. For Scenario $2$, $\mathcal{V}$ is generated randomly with equal probability. In Scenario $2$, Table \ref{tabmy_label} lists the theoretical probabilities that the signal $\mathbf{s}$ can be uniquely recovered from modulo samples $\mathbf{z}$ for different values of $N$.
\begin{table}[htb!]
    \begin{center}
    \caption{Theoretical probabilities for unique recovery of $\mathbf{s}$ \\ in Scenario $2$.}\label{tabmy_label}
        \begin{tabular}{|c|c|c|c|c|c|c|c|c|c|}
            \hline
             N&5&6&7&8&10&11&14&16  \\ \hline
             Prob.&0.47&0.14&0.49&0.04&0.22&0.50&0.24&0.03\\ \hline
        \end{tabular}
    \end{center}
\end{table}
In Fig. \ref{Simu1}, the probabilities of successful recovery for Scenario $1$ and Scenario $2$ for different values of $N$ using the proposed recovery algorithm are presented. For Scenario $1$, the probabilities of successful recovery for all $N$ are equal to $1$. In contrast, Scenario $2$ shows a noticeable decrease in the probabilities of successful recovery compared to Scenario $1$. It is worth noting that the probabilities of successful recovery for Scenario $2$ are significantly higher than the theoretical results listed in Table \ref{tabmy_label}. The main reason for this discrepancy is that we constrained the dynamic range of the variables $\boldsymbol{\epsilon}$ (\ref{equationsolver3}) in the recovery algorithm to make the running time acceptable. Under these conditions, some models that are unidentifiable with an unlimited dynamic range of variables $\boldsymbol{\epsilon}$ become identifiable with a limited dynamic range.

\begin{figure}[htb!]
    \centering
    \includegraphics[width=0.4\textwidth]{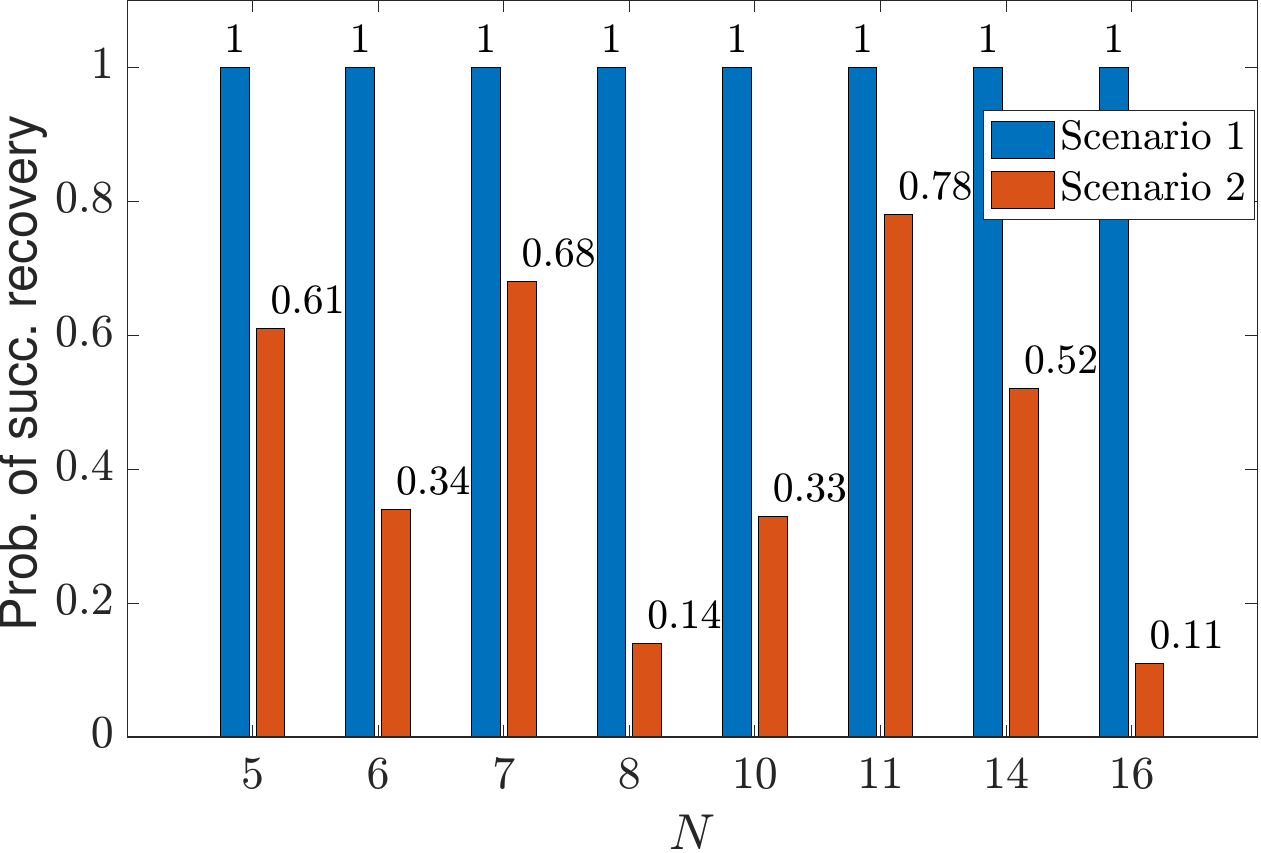}
    \caption{Probabilities of successful recovery of the proposed recovery algorithm over $300$ Monte Carlo trials.}
    \label{Simu1}
\end{figure}

\subsection{A necessary condition for Vandermonde sensing matrices}
Let $\alpha_0,\alpha_1,\cdots,\alpha_{N-1}\in\mathbb{C}$ be the nodes which are used to generate the Vandermonde matrix $\mathbf{V}\in\mathbb{C}^{N\times N}$, i.e., the $(i,j)$th element of $\mathbf{V}$ is $\alpha_j^i$, $i=0,1,\cdots,N-1$ and $j=0,1,\cdots,N-1$. Note that the DFT matrix multiplied by $\sqrt{N}$, i.e., $\sqrt{N}\mathbf{F}$, is a special case of a Vandermonde matrix where $\alpha_j = {\rm e}^{-{\rm j}\frac{2\pi j}{N}}, j=0,1,\cdots,N-1$. In this subsection, we provide a necessary condition for uniquely identifying the original signal from modulo measurements under Vandermonde sensing matrix. Similarly to model (\ref{decomsimple}) under DFT sensing matrix, the measurement model is 
\begin{align}\label{moduloequationVan}
\mathbf{z}=\mathbf{y}+\boldsymbol{\epsilon}=\mathbf{V}\mathbf{s}+\boldsymbol{\epsilon},
\end{align}
where the original signal $\mathbf{s}$ is constrained to $\mathbb{C}^{N}_{\mathcal{V}}$. Let $\mathcal{V}_{r}$ denote the set of all indices $j$ such that the nodes $\alpha_j$ is a Gaussian rational number for $j\in\mathcal{N} = \{0,1,\cdots,N-1\}$, that is, 
\begin{align}\label{DefSr}
    \mathcal{V}_{r}\triangleq\{j \mid \alpha_j\in\mathbb{Q}[{\rm j}], j\in\mathcal{N}\}.
\end{align}
A necessary condition to uniquely recover $\mathbf{s}$ from $\mathbf{z}$ is introduced in Theorem \ref{TheoNece}.
\begin{theorem}\label{TheoNece}
    If $\mathbf{s}\in\mathbb{C}^N_{\mathcal{V}}$ can be uniquely recovered from $\mathbf{z}$ in the measurement model (\ref{moduloequationVan}), then we have $\mathcal{V}_r\subseteq \mathcal{V}$, where $\mathcal{V}_r$ is the  set defined in (\ref{DefSr}) which records all indices $j$ such that the nodes $\alpha_j$ of $\mathbf{V}$ is a Gaussian rational number.
\end{theorem}
\begin{proof}
We prove the contrapositive of this proposition. Let $i$ be an element that satisfies $i\in\mathcal{V}_r$ and $i\notin\mathcal{V}$. Let 
\begin{align}\label{defth2setA}
    \mathcal{A}_i\triangleq\{\alpha_i^n\mid n\in\mathcal{N}\},
\end{align}
which is the set that contains all elements of the $i$th column of $\mathbf{V}$. According to the definition of $\mathcal{V}_r$, we have that each element in $\mathcal{A}_i$ is a Gaussian rational number. Therefore, there exists a nonzero integer $I$ such that $I\mathcal{A}_i$ is a subset of Gaussian integers, i.e., $I\mathcal{A}_i\subseteq\mathbb{Z}[{\rm j}]$. Let $\mathbf{s}' = \mathbf{s} + I\mathbf{e}_i$ where $\mathbf{e}_i$ is the one-hot vector with the $i$th element being $1$ and the others being $0$. Note that we also have $\mathbf{s}'\in\mathbb{C}^N_{\mathcal{V}}$ as $i\notin\mathcal{V}$. Furthermore, let $\mathbf{y}' = \mathbf{V}\mathbf{s}'$ be the unfolded signal generated by $\mathbf{s}'$. We have $\mathbf{y}' = \mathbf{V}(\mathbf{s} + I\mathbf{e}_{i}) = \mathbf{y} + I\mathbf{V}_{i}$ based on the fact $\mathbf{y} = \mathbf{V}\mathbf{s}$, where $\mathbf{V}_i$ is the $i$th column of $\mathbf{V}$. Given that $\mathcal{A}_i$ (\ref{defth2setA}) is the set that contains all elements of the $i$th column of $\mathbf{V}$ and $I\mathcal{A}_i\subseteq\mathbb{Z}[{\rm j}]$, we conclude that $I\mathbf{V}_{i}\in\mathbb{Z}[{\rm j}]^N $ is a vector of Gaussian integers. Let $\mathbf{z}' = \mathscr{C}(\mathbf{y}')$ be the modulo measurements of $\mathbf{y}'$. Thus, we have $\mathbf{z}' = \mathscr{C}(\mathbf{y} + I\mathbf{V}_{i})= \mathscr{C}(\mathbf{y})$, which is equal to $\mathbf{z}$.
\end{proof}

A natural question is whether the necessary condition in Theorem \ref{TheoNece} can also serve as a sufficient condition. Unfortunately, the answer is negative. In the following, a counterexample is presented to demonstrate compliance with the necessary condition while remaining unidentifiable.
\newline
\emph{Counterexample}: We consider the case where $N=8$ and $\mathbf{V} = 2\sqrt{2}\mathbf{F}$, where $\mathbf{F}\in\mathbb{C}^{8\times 8}$ is a DFT matrix. The nodes of $\mathbf{V}$ are $\{\alpha_j = {\rm e}^{{\rm j}\frac{2\pi j}{8}}\}_{j=0}^{7}$, and $\alpha_0$, $\alpha_2$, $\alpha_4$, and $\alpha_6$ are Gaussian rational numbers. Thus, we have $\mathcal{V}_r=\{0,2,4,6\}$. Let $\mathcal{V} = \{0,2,3,4,6,7\}$ which satisfies the necessary condition in Theorem \ref{TheoNece}, i.e., $\mathcal{V}_r\subseteq\mathcal{V}$. We construct another vector $\mathbf{s}'\in \mathbb{C}^N_{\mathcal{V}}$ shown as 
\begin{align}
    s'_i = \begin{cases}
        0, \text{if } i\in{\mathcal{V}}~(i=0,2,3,4,6,7)\\
        s_i +  2\sqrt{2}, \text{if } i\in{\bar{\mathcal{V}}}~(i=1,5)
    \end{cases}.
\end{align}
Let $\mathbf{y}' = \mathbf{V}\mathbf{s}'$ be the unfolded measurement generated by $\mathbf{s}'$. Then, we have $\mathbf{y}' = \mathbf{F}\mathbf{s}+2\sqrt{2}\mathbf{F}_{1} + 2\sqrt{2}\mathbf{F}_{5} = \mathbf{y} + [2,0,-2{\rm j},0,-2,0,2{\rm j},0]^{\rm T}$. Let $\mathbf{z}' = \mathscr{C}(\mathbf{y}')$ be the modulo measurements of $\mathbf{y}'$. We have $\mathbf{z}' =  \mathscr{C}(\mathbf{y} + [2,0,-2{\rm j},0,-2,0,2{\rm j},0]^{\rm T}) = \mathscr{C}(\mathbf{y}) = \mathbf{z}$. Therefore, the model is unidentifiable.

\subsection{Proof of Theorem \ref{TheoComplex}}\label{proofTheo1}
Before proving Theorem \ref{TheoComplex}, we first introduce Lemma \ref{lemmaT2_1}. Let $f(x)$ be a polynomial in $\mathbb{Q}[{\rm j}][x]$ of degree $N$, where $\mathbb{Q}[{\rm j}][x]=\{\sum_{i=0}^{n}a_ix^i\mid a_i\in\mathbb{Q}[{\rm j}],n\in\mathbb{N}\}$ is the ring of polynomials over $\mathbb{Q}[{\rm j}]$, and let $\mathcal{R}_f\triangleq \{x\mid f(x) = 0\}$ be the set of roots of the polynomial $f(x)$. Lemma \ref{lemmaT2_1} shows that $f(x) = A(x^N-1)$, where $A$ is any nonzero Gaussian rational number if and only if $\forall k\in\{1,\cdots,K\}$, $\exists x\in\mathcal{R}_f$ such that $c_k(x) = 0$, where $c_k(x)$ is the $k$th factor of $C(x)= x^N-1$ defined in (\ref{factorCx}).

\begin{lemma}\label{lemmaT2_1}
$f(x) = A(x^N-1)$, where $A\in\mathbb{Q}[{\rm j}]$ is any nonzero Gaussian rational number, holds for all $x\in\mathbb{C}$ if and only if $\forall k\in\{1,\cdots,K\}$, $\exists x\in\mathcal{R}_f$ such that $c_k(x) = 0$.
\end{lemma}
\begin{proof}
We first prove the necessary part, which can be easily obtained as $f(x) = A(x^N-1) = AC(x) = Ac_1(x)\cdots c_{K}(x)$ according to the fact $C(x) = c_1(x) \cdots c_K(x)$ (\ref{factorCx}). 

For the sufficient part, let $x_1,x_2,\cdots,x_K\in\mathcal{R}_f$ be the particular roots of $c_1(x),c_2(x),\cdots,c_K(x)$, respectively, meaning $c_k(x_k)=0$ for $k=1,\cdots,K$. Since $C(x) = x^N-1 = c_1(x) \cdots c_K(x)$ (\ref{factorCx}) does not have repeated roots, it follows that $x_i\neq x_j$ when $i\neq j$. According to the polynomial remainder theorem in Proposition \ref{polyReTh}, the division of $f(x)$ by $c_1(x)$ yields 
\begin{align}\label{Th2lemma1decom}
    f(x) = c_1(x)q_1(x) + r_1(x),
\end{align}
where $q_1(x)$ and $r_1(x)$ are polynomials in $\mathbb{Q}[{\rm j}][x]$ and the degree of $r_1(x)$ is less than that of $c_1(x)$. Recall that $x_1$ is a root of both $f(x)$ and $c_1(x)$, so we have $f(x_1) = c_1(x_1) = 0$. Thus, according to (\ref{Th2lemma1decom}), we have $r_1(x_1) = 0$. Since $c_1(x)$ is irreducible over $\mathbb{Q}[{\rm j}]$, based on Proposition \ref{CycPriP}, it follows that there is no polynomial in $\mathbb{Q}[{\rm j}][x]$ with degree less than that of $c_1(x)$ that has $x_1$ as a root. Therefore, the polynomial $r_1(x)$ must be zero. Consequently, based on (\ref{Th2lemma1decom}), we have 
\begin{align}\label{Th1lemma1fcq}
    f(x) = c_1(x)q_1(x).
\end{align}
Based on the equation $N = \sum_{k=1}^K{\rm deg}(c_k(x))$ obtained from $C(x) = x^N-1 = c_1(x) \cdots c_K(x)$, we have that ${\rm deg}(q_1(x)) = N - {\rm deg}(c_1(x))\geq {\rm deg}(c_2(x))$. Similarly to (\ref{Th2lemma1decom}), according to the polynomial remainder
theorem in Proposition \ref{polyReTh}, the division of $q_1(x)$ by $c_2(x)$ yields 
\begin{align}\label{Th2lemma1b}
    q_1(x) = c_2(x) q_2(x) + r_2(x),
\end{align}
where $q_2(x)$ and $r_2(x)$ are polynomials in $\mathbb{Q}[{\rm j}][x]$. As $f(x) = c_1(x)q_1(x)$ (\ref{Th1lemma1fcq}) and $x_2$ is a root of $f(x)$, we have $f(x_2) = c_1(x_2)q_1(x_2)=0$. In addition, we have $c_1(x_2)\neq 0$ since $x_2$ is the root of $c_2(x)$ and $C(x)$ does not have repeated roots. Therefore, we have $q_1(x_2) = 0$. Furthermore, we conclude that $r_2(x_2) = 0$ according to (\ref{Th2lemma1b}). Similarly to $c_1(x)$, there does not exist a polynomial in $\mathbb{Q}[{\rm j}][x]$ of degree less than $c_2(x)$ having a root at $x_2$ according to Proposition \ref{CycPriP}. Thus, the polynomial $r_2(x)$ is equal to $0$. Therefore, we obtain $q_1(x) = c_2(x)q_2(x)$ according to (\ref{Th2lemma1b}), and thus $f(x) = c_1(x)c_2(x)q_2(x)$ according to (\ref{Th2lemma1decom}). Continuing this process, we can finally express $f(x)$ as $f(x) = c_1(x)c_{2}(x)\cdots c_{K}(x)q_{K}(x)$. According to the equation $C(x) = x^N-1=c_1(x) \cdots c_K(x)$ (\ref{factorCx}), we have $f(x) = (x^N-1)q_K(x)$. Since $f(x)$ is an $N$-degree polynomial in $\mathbb{Q}[{\rm j}][x]$, we conclude that $q_K(x)$ is a constant. Therefore, we have $f(x) = A(x^N-1)$, where $A$ can be any nonzero number in $\mathbb{Q}[{\rm j}]$.
\end{proof}

Here, we reintroduce Theorem \ref{TheoComplex}: \emph{In the modulo-DFT sensing model (\ref{decomsimple}), the signal $\mathbf{s}\in\mathbb{C}^N_{\mathcal{V}}$ can be uniquely recovered from the modulo measurements $\mathbf{z}$ if and only if $\forall k\in\{1,\cdots,K\}$, $\mathcal{N}_k\cap\mathcal{V}\neq \emptyset$.} 

The proof of Theorem \ref{TheoComplex} is provided below, which is based on Lemma \ref{lemmaT2_1}.
\newline
\textbf{\emph{Proof of Theorem \ref{TheoComplex}:}} According to Proposition \ref{EquiCondi}, proving Theorem \ref{TheoComplex} is equivalent to proving: \emph{There does not exist a nonzero Gaussian integer polynomial $f(x)$ of maximum degree $N-1$ satisfying $f({\rm e}^{{\rm j}\frac{2\pi n}{N}}) = 0$ for all $n\in\mathcal{V}$ if and only if $\forall k\in\{1,\cdots,K\}$, $\mathcal{N}_k\cap\mathcal{V}\neq \emptyset$.}

We first prove the necessary part by proving the contrapositive, i.e, if $\exists k\in\{1,\cdots,K\}$ such that $\mathcal{N}_k\cap\mathcal{V}= \emptyset$, then there exists a nonzero Gaussian integer polynomial $f(x)$ of maximum degree $N-1$ satisfying $f({\rm e}^{{\rm j}\frac{2\pi n}{N}}) = 0$ for all $n\in\mathcal{V}$. Let $i$ be an integer in $\{1,2,\cdots, K\}$ such that $\mathcal{N}_i\cap\mathcal{V} = \emptyset$. When $K=1$, we have $i=1$ and $\mathcal{N}_i = \mathcal{N}$. Thus, we have $\mathcal{N} \cap \mathcal{V} = \emptyset$. Since $\mathcal{V} \subseteq \mathcal{N}$, we have $\mathcal{V} = \emptyset$. Therefore, any Gaussian integer polynomials $f(x)$ of maximum degree $N-1$, such as  $f(x) = A$, where $A \in \mathbb{Z}[\rm j]$, satisfies $f({\rm e}^{{\rm j}\frac{2\pi n}{N}}) = 0$ for all $n\in\mathcal{V}$. In the following, we consider the case that $K\geq 2$. Let $M(x)$ be a polynomial defined as  
\begin{align}\label{defMx}
    M(x) \triangleq \prod_{k=1,k\neq i}^Kc_k(x),
\end{align}
where $\{c_k(x)\}_{k=1}^K$ are factors of $C(x) = x^N-1$ (\ref{factorCx}) and are all polynomials with Gaussian rational coefficients. Therefore, $M(x)$ is a nonzero polynomial with Gaussian rational coefficients of maximum degree $N-1$. Consequently, there exists an integer $B$ such that $f(x) = B M(x)$ is a nonzero polynomial with Gaussian integer coefficients of maximum degree $N-1$. Below, we prove that $f(x)$ satisfies $f({\rm e}^{{\rm j}\frac{2\pi n}{N}}) = 0$ for all $n \in \mathcal{V}$. Since $\{c_k(x)\}_{k=1,k\neq i}^{K}$ are all factors of $M(x)$ (\ref{defMx}) and also $f(x) = B M(x)$, we have $\bar{\mathcal{C}_i} \triangleq \bigcup_{k=1,k\neq i}^K \mathcal{C}_k = \mathcal{C}\backslash\mathcal{C}_i$ are roots of $f(x)$, where $\mathcal{C}_k$ is the set of roots of $c_k(x)$. Furthermore, recall that $\mathcal{N}_k$ (\ref{defni}) is the set that records the powers of ${\rm e}^{{\rm j}\frac{2\pi}{N}}$ in $\mathcal{C}_k$. We have $f({\rm e}^{{\rm j}\frac{2\pi n}{N}}) = 0$ for all $n \in \bar{\mathcal{N}_i} \triangleq \mathcal{N} \backslash \mathcal{N}_i$. Since $\mathcal{N}_i \cap \mathcal{V} = \emptyset$, it follows that $\mathcal{V} \subseteq \bar{\mathcal{N}_i}$. Therefore, $f({\rm e}^{{\rm j}\frac{2\pi n}{N}}) = 0$ for all $n \in \mathcal{V}$.

Now, we prove the sufficient part, i.e., if $\forall k\in\{1,\cdots,K\}$, $\mathcal{N}_k\cap\mathcal{V}\neq \emptyset$, then there does not exist a nonzero Gaussian integer polynomial $f(x)$ of maximum degree $N-1$ satisfying $f({\rm e}^{{\rm j}\frac{2\pi n}{N}}) = 0$ for all $n\in\mathcal{V}$. We prove it by contradiction. Assume there exists a nonzero Gaussian integer polynomial $f(x)$ of maximum degree $N-1$ satisfying $f({\rm e}^{{\rm j}\frac{2\pi n}{N}}) = 0$ for all $n\in\mathcal{V}$. Let $g(x) \triangleq x^{N - \deg(f(x))} f(x)$, which is a polynomial with Gaussian integer coefficients of degree $N$ and also satisfies $g({\rm e}^{{\rm j}\frac{2\pi n}{N}}) = 0$ for all $n \in \mathcal{V}$. Since ${\rm deg}(f(x))\leq N-1$, $x$ is a factor of $g(x)$ and $g(0)=0$. Considering $\forall k\in\{1,\cdots,K\}$, $\mathcal{N}_k\cap\mathcal{V}\neq \emptyset$ and $g({\rm e}^{{\rm j}\frac{2\pi n}{N}}) = 0$ for all $n\in\mathcal{V}$, we have that $\forall k\in\{1,\cdots,K\}$, $\exists n\in\mathcal{N}_k$ such that $g({\rm e}^{{\rm j}\frac{2\pi n}{N}}) = 0$. In addition, since $\mathcal{N}_k$ (\ref{defni}) is the set that records the powers of ${\rm e}^{{\rm j}\frac{2\pi}{N}}$ in $\mathcal{C}_k$, we have $\forall k \in \{1, \cdots, K\}$, $\exists x \in \mathcal{C}_k$ such that $g(x) = 0$. This implies that $\forall k \in \{1, \cdots, K\}$, there exists a root of $g(x)$ which is also a root of $c_k(x)$. Thus, according to Lemma \ref{lemmaT2_1}, we have $g(x) = A(x^N-1)$, where $A$ is any nonzero number in $\mathbb{Z}[{\rm j}]$, and $g(0) = -A\neq 0$. However, this contradicts the previously derived result $g(0)=0$ as $x$ is a factor of $g(x)$.

\section{Identifiability of Periodic Bandlimited Signals}
In this section, we discuss the identifiability of PBL signals under MS which can be considered as the modulo-DFT sensing model with additional symmetric and conjugate constraints on the original signal when the sampling rate is greater than the Nyquist rate, and we provide a necessary and sufficient condition based solely on the number of measurements $N$ to uniquely identify the PBL signal from the modulo samples up to a constant factor.
\subsection{Problem setup}
A PBL signal with $P\geq 1$ harmonics is
\begin{align}\label{periodicModel1}
    g(t) = \sum\nolimits_{p=-P}^Pc_p{\rm e}^{{\rm j}2\pi pf_0 t},
\end{align}
where $c_p$ is the Fourier series coefﬁcient and $c_p = c^*_{-p}$, $p=0,1,\cdots,P$. 
In the context of MS with the sampling rate $f_s = Nf_0$, $N$ samples are obtained in one period which are 
\begin{align}\label{PBSmoduloModel}
    z_n = \mathscr{M}(y_n),~n=0,1,\cdots,N-1,
\end{align}
where 
\begin{align}
    y_n = g(nT_s)= \sum_{p=-P}^Pc_p{\rm e}^{{\rm j}\frac{2\pi np}{N}}
\end{align}
are unfolded samples, $\mathscr{M}(\cdot)$ is the modulo operator defined in (\ref{modulooperator}), and $T_s = 1/f_s$ is the sampling period. Similar to (\ref{decomsimple}), let $\mathbf{z} \triangleq [z_0,z_1,\cdots,z_{N-1}]^{\rm T}$ and $\mathbf{y} \triangleq [y_0,y_1,\cdots,y_{N-1}]^{\rm T}$, we have 
\begin{align}\label{PBSmoduloDecom}
    \mathbf{z} = \mathbf{y} + \boldsymbol{\epsilon},
\end{align}
where $\boldsymbol{\epsilon}\in\mathbb{Z}^N$ is an integer vector of length $N$.

For the PBL signal (\ref{periodicModel1}), when the number of samples in one period is less than or equal to $2P+1$, i.e., $N\leq 2P+1$, without considering the ambiguity on the DC component, i.e., $c_0$, the PBL signal can not be uniquely recovered from the modulo samples. This conclusion has been proved in \cite[Theorem 2]{mulleti2024modulo}. It is worth noting that while the PBL signal can be uniquely recovered from the unfolded samples $\mathbf{y}$ in the conventional sampling scheme when $N = 2P + 1$, the introduction of the modulo operation leads to an ambiguity that makes the PBL signal unidentifiable under MS for the same value of $N$. In this paper, we study the identifiability of the PBL signal (\ref{periodicModel1}) under MS when the number of samples in one period is greater than $2P+1$, i.e., $N> 2P+1$. In this scenario, the measurement model (\ref{PBSmoduloDecom}) can be further written as 
\begin{align}\label{periodicModel}
    \mathbf{z} = \mathbf{y} + \boldsymbol{\epsilon} = \mathbf{F}\mathbf{s} + \boldsymbol{\epsilon}
\end{align}
where $\mathbf{F}\in\mathbb{C}^{N\times N}$ is the DFT matrix, and 
\begin{align}
    {\mathbf{s}}= \sqrt{N}[c_0,c_{-1},\cdots,c_{-P},0,\cdots,0,c_{P},\cdots,c_1]^{\rm T}
\end{align}
is a complex vector generated by the Fourier series coefficients $c_{-P}, \cdots, c_0, \cdots, c_{P}$. We observe that the PBL signal under MS (\ref{periodicModel}) is also a modulo-DFT sensing model (\ref{moduloequation}) when $N> 2P+1$. However, compared to the general modulo-DFT sensing model (\ref{moduloequation}), several additional constraints are imposed on the signal $s$ for the measurement (\ref{periodicModel}). In detail, $\mathbf{s}$ is required to satisfy the symmetric and conjugate property, meaning $s_{n} = s^*_{N-n}$ for $n = 1, \cdots, N-1$ and $s_0 \in \mathbb{R}$ is a real number, which ensures that the modulo samples in the model (\ref{periodicModel}) are real numbers. In addition, given the number of measurements $N$, the set $\mathcal{P}$, which records the indices of the zero elements of $\mathbf{s}$ in the model (\ref{periodicModel}), is in the form of 
\begin{align}\label{Pdefine}
    \mathcal{P} = \{P+1, P+2, \cdots, N-P-1\},
\end{align}
whereas the set $\mathcal{V}$, which records the indices of the zero elements of $\mathbf{s}$ in the model (\ref{moduloequation}), is arbitrary. Therefore, for the PBL signal under MS (\ref{periodicModel}), the equivalent signal $\mathbf{s}$ belongs to the constraint set $\mathbb{S}^{N}_{\mathcal{P}}$ which is defined as 
\begin{align}
    \mathbb{S}^{N}_{\mathcal{P}}
    \triangleq\{\mathbf{s} \mid \mathbf{s}\in\mathbb{C}^N_{\mathcal{P}}, s_0\in\mathbb{R}, s_{n} = s^*_{N-n}, n=1,2,\cdots,N-1\}.\notag
\end{align} 
The set $\mathbb{S}^{N}_{\mathcal{P}}$ is a subspace of $\mathbb{C}^N_{\mathcal{P}}$ which is also closed under addition. In the following, we investigate the identifiability of PBL signals under MS (\ref{PBSmoduloModel}) when $N > 2P+1$, which is equivalent to the identifiability of the modulo-DFT sensing model (\ref{periodicModel}) for all $\mathbf{s}\in\mathbb{S}^{N}_{\mathcal{P}}$. 

\subsection{Identifiability}
A necessary and sufficient condition for uniquely recovering $\mathbf{y}$ from $\mathbf{z}$ up to an integer constant vector is introduced in Theorem \ref{theorem3}, where 
\begin{align}\label{DefD_N}
    \mathcal{D}\triangleq\left\{n\mid n\in\mathbb{Z}^{+},n\mid N\right\}
\end{align}
is the set that contains all positive integers that divide $N$ and $\Phi_{d}(x)$ is the $d$th cyclotomic polynomial. To make this paper self-contained, we introduce cyclotomic polynomials and their properties in Appendix \ref{AppenD}.
\begin{theorem}\label{theorem3}
     The samples of the PBL signal, i.e., $\mathbf{y}$, are uniquely identifiable up to an integer constant vector from modulo samples $\mathbf{z}$ if and only if $\forall d\in\mathcal{D}\setminus\{1\}, \exists n\in\mathcal{P}$ such that ${\rm e}^{{\rm j}\frac{2\pi n}{N}}$ is a root of $\Phi_{d}(x)$.
\end{theorem}
\begin{proof}
    The proof is postponed to Sec. \ref{ProofPBSTh3}.
\end{proof}

\emph{Remark:} The necessary and sufficient condition for the unique identification of PBL signals under MS (\ref{periodicModel}) shown in Theorem \ref{theorem3} has a close relationship with the factorization of $C(x) = x^N - 1$ over the rational number field $\mathbb{Q}[\text{j}]$. According to Property \ref{FacC} and Property \ref{irre_poly} in Appendix \ref{AppenD}, $C(x)$ is factored as 
\begin{align}\label{CxfactorPBS}
    C(x) = x^N - 1 = \prod_{d \in \mathcal{D}} \Phi_d(x)
\end{align}
over the rational number field $\mathbb{Q}$. Therefore, the condition that $\forall d \in \mathcal{D} \setminus \{1\}, \exists n \in \mathcal{P}$ such that $\text{e}^{\text{j} \frac{2\pi n}{N}}$ is a root of $\Phi_d(x)$ in Theorem \ref{theorem3} is equivalent to the condition that for any factor of $C(x)$ over the rational number field $\mathbb{Q}$, except $\Phi_{1}(x) = x - 1$, there exists $n \in \mathcal{P}$ such that $\text{e}^{\text{j} \frac{2\pi n}{N}}$ is a root of that factor. Recall that for the unique identification of the modulo-DFT sensing model (\ref{decomsimple}) shown in Corollary \ref{CoroTheorem}, which is a corollary of Theorem \ref{TheoComplex}, the necessary and sufficient condition for unique identification is equivalent to $\forall k \in \{2, 3, \cdots, K\}, \exists n \in \mathcal{V}$ such that $\text{e}^{\text{j} \frac{2\pi n}{N}}$ is a root of $c_k(x)$, where $c_k(x)$ is the $k$th factor of the factorization of $C(x) = x^N - 1$ over the Gaussian rational number field $\mathbb{Q}[\text{j}]$ (\ref{factorCx}), according to the definition of $\mathcal{N}_k$ (\ref{defni}).\footnote{In Corollary \ref{CoroTheorem}, the original signal $\mathbf{s}_{1:N-1}$ can be uniquely recovered from $\mathbf{z}$ is equivalent to the unfolded samples $\mathbf{y}$ being uniquely recovered from $\mathbf{z}$ up to an integer constant vector.} Therefore, the only difference between the necessary and sufficient condition for unique identification in Theorem \ref{theorem3} and that in Corollary \ref{CoroTheorem} is that for PBL signals under MS (\ref{periodicModel}), the condition is derived from the factorization of $C(x) = x^N - 1$ over the rational number field $\mathbb{Q}$ since the modulo samples are real numbers, whereas for the modulo-DFT sensing model (\ref{decomsimple}), the condition is derived from the factorization of $C(x) = x^N - 1$ over the Gaussian rational number field $\mathbb{Q}[\text{j}]$ since the modulo samples are complex numbers.

Compared to $\mathcal{V}$, the set records the indices of the zero elements of the original signal $\mathbf{s}$ in the modulo-DFT sensing model (\ref{decomsimple}), which can be any subset of $\mathcal{N} = \{0,1,\cdots,N-1\}$, the corresponding set $\mathcal{P}$ (\ref{Pdefine}) for PBL signals under MS (\ref{periodicModel}) has a specific structure and is determined by $P$ and $N$. Surprisingly, after further study, we discovered that the necessary and sufficient condition for the unique identification in Theorem \ref{theorem3} can be transformed into a condition only with respect to the number of measurements $N$, and the result is illustrated in Lemma \ref{PBSlemma3}, where $H(N)$ is defined as 
\begin{align}\label{defineHN}
    H(N) = \frac{N}{\operatorname{max}_{d\in\mathcal{D}\setminus\{1\}}h(d)}.
\end{align}
Here, $h(d)$ is a function of integers greater than $1$ defined as
\begin{align}\label{definehn}
    h(d) = \begin{cases}
    2+\frac{2}{d-1},~2\nmid d;&({\rm Case~} 1)\\
    2,~d=2;&({\rm Case~} 2)\\
    2+\frac{8}{d-4},~2\mid d, d\neq 2\text{ and } 4\nmid d;&({\rm Case~} 3)\\
    2+\frac{4}{d-2},~4\mid d;&({\rm Case~} 4)\\
    \end{cases},
\end{align}
where $a\nmid b$ means that $a$ does not divide $b$. Therefore, based on Theorem \ref{theorem3} and Lemma \ref{PBSlemma3}, we have that the unfoled samples $\mathbf{y}$ are uniquely identifiable up to an integer constant vector from modulo samples $\mathbf{z}$ if and only if $H(N) \geq P+1$, and the conclusion is summarized in Theorem \ref{theorem2}.
\begin{lemma}\label{PBSlemma3}
    $\forall d\in\mathcal{D}\setminus\{1\}, \exists n\in\mathcal{P}$ such that ${\rm e}^{{\rm j}\frac{2\pi n}{N}}$ is a root of $\Phi_{d}(x)$ if and only if $H(N)\geq P+1$.
\end{lemma}
\begin{proof}
    The proof is postponed to Sec. \ref{ProofTh3}.
\end{proof}
\begin{theorem}\label{theorem2}
    The samples of the PBL signal, i.e., $\mathbf{y}$, are uniquely identifiable up to an integer constant vector from modulo samples $\mathbf{z}$ if and only if $H(N) \geq P+1$.
\end{theorem}

According to Theorem \ref{theorem2}, the identifiability of model (\ref{periodicModel}) depends on $H(N)$ which is a function of $N$ and its factors. Directly calculating $H(N)$ is intuitively difficult. In the following, some specific examples are introduced.

\textbf{Case \RomanNumeralCaps{1}:} Similar to Sec. \ref{SecIdenti}, we first consider the case that $N = 2^M$ where $M\in\mathbb{Z}$ and $M\geq 2$. In this case, $\mathcal{D} = \{1,2,2^2,\cdots,2^M\}$. Thus, $\operatorname{max}_{d\in\mathcal{D}\setminus\{1\}}h(d) = 4$ and $H(N) = N/4 = 2^{M-2}$. Therefore, the sufficient and necessary condition is $2^{M-2}\geq P+1$, and the conclusion is summarized in Corollary \ref{PBScorollary1}.
\begin{corollary}\label{PBScorollary1}
When $N = 2^M$ where $M\in\mathbb{Z}$ and $M\geq 2$, then $\mathbf{y}$ are uniquely identifiable up to an integer constant vector from modulo samples $\mathbf{z}$ if and only if $2^{M-2} = N/4 \geq P+1$.
\end{corollary}
    
\textbf{Case \RomanNumeralCaps{2}:} We now consider the case where $N$ is an odd number. Let $a$ be the smallest factor of $N$ other than $1$. In this case, all the elements in $\mathcal{D}$ are odd numbers. Thus, $\operatorname{max}_{d\in\mathcal{D}\setminus\{1\}}h(d) = 2+\frac{2}{a-1}$ and $H(N) = N/(2+\frac{2}{a-1}) = \frac{N(a-1)}{2a}$. Therefore, the sufficient and necessary condition is $\frac{N(a-1)}{2a}\geq P+1$. The result is summarized in Corollary \ref{PBScorollary2}. 
\begin{corollary}\label{PBScorollary2}
When $N$ is an odd number, and $a$ is the smallest factor of $N$ other than $1$, then $\mathbf{y}$ are uniquely identifiable up to an integer constant vector from modulo samples $\mathbf{z}$ if and only if $\frac{N(a-1)}{2a}\geq P+1$.
\end{corollary}
In \cite{mulleti2024modulo}, the case that $N$ is a prime number is discussed, which can be viewed as a special case of Case \RomanNumeralCaps{2} when $N\neq 2$. In this case, the necessary and sufficient condition is $N\geq 2P+3$ since $N$ only has two factors $1$ and $N$. Therefore, when $P$ is known, setting the number of measurements in one period as a prime number greater than $2P+1$, the PBL signal under MS (\ref{PBSmoduloModel}) is identifiable up to an integer constant. However, when $P$ becomes large, finding a prime number larger than $2P+1$ becomes increasingly difficult as the primes become farther apart. In this paper, according to Corollary \ref{PBScorollary2}, a more general result is obtained. In detail, given a prime number $a > 2$, when the number of measurements $N \geq \frac{2a}{a-1}(P+1)$ and the smallest factor of $N$ (other than $1$) is $a$, then the PBL signal under MS (\ref{PBSmoduloModel}) is identifiable up to an integer constant. This conclusion provides a trade-off in determining the number of samples $N$. On one hand, when $a$ is small, it is easy to find an integer with the minimum factor (other than $1$) greater than or equal to $a$, while the lower bound of the number of samples, i.e., $\frac{2a}{a-1}(P+1)$ is relatively large. On the other hand, when $a$ is large, the lower bound of the number of samples, i.e., $\frac{2a}{a-1}(P+1)$ is relatively small, while finding an integer with the minimum factor (other than $1$) greater than or equal to $a$ is difficult. 

In the following, we consider a general conclusion on the identifiability of PBL signals with respect to the oversampling factor. Given a PBL signal defined in (\ref{periodicModel1}), the Nyquist frequency is $f_{\rm Nyq}=2Pf_0$. Thus, the oversampling factor is $\gamma = \frac{f_s}{f_{\rm Nyq}} = \frac{N}{2P}$, where $f_s = Nf_0$ is the sampling rate. Therefore, the necessary and sufficient condition in Theorem \ref{theorem2}, that is, $H(N)\geq P+1$, is equivalent to $\gamma \geq \frac{1}{2}\operatorname{max}_{d\in\mathcal{D}\setminus\{1\}}h(d)(1+\frac{1}{P})$. We first consider the maximum value of function $h(d)$ where $d\in\mathbb{Z}$ and $d\geq 2$. Note that $h(d)$ decreases monotonically with $d$ in Case $1$, Case $3$, and Case $4$. We have $\max_{d\in\mathbb{Z}, d\geq 2}h(d) = \max\{h(3),h(2),h(6),h(4)\} = \max\{3,2,6,4\} = 6$. In summary, the maximum value of the function $h(d)$ is $6$ at $d=6$. Thus, we have $\frac{N}{6}\geq H(N)$ for any $N$ where $N\geq 4$. Additionally, with Theorem \ref{theorem2} and the inequality $\frac{N}{6}\geq H(N)$, we can deduce that if $N \geq 6(P+1)$, the samples of the PBL signal, i.e., $\mathbf{y}$, are uniquely identifiable up to an integer constant vector from modulo samples $\mathbf{z}$. Given $\gamma =\frac{N}{2P}$, it can be further deduced that when $\gamma\geq3(1+\frac{1}{P})$, then $\mathbf{y}$ are uniquely identifiable up to an integer constant vector from modulo samples $\mathbf{z}$, and the result is shown in Corollary \ref{PBScorollary3}.
\begin{corollary}\label{PBScorollary3}
    If $\gamma\geq3(1+\frac{1}{P})$, then the samples of the PBL signal, i.e., $\mathbf{y}$, are uniquely identifiable up to an integer constant vector from modulo samples $\mathbf{z}$. 
\end{corollary}
\subsection{Simulation}
In this subsection, we present the probabilities of successful recovery of PBL signals from modulo samples over $300$ Monte Carlo trials. Similar to the method utilized in Sec. \ref{Simulation1}, we also use the professional optimization software Gurobi to solve the integer linear equations (\ref{equationsolver3}) and estimate the unfolded samples $\mathbf{y}$. Note that for PBL signals, $\mathbf{z}$ and $\boldsymbol{\epsilon}$ are integer vectors instead of Gaussian integer vectors. Both the real and imaginary parts of the nonzero elements of $c_p, p = 1, 2, \cdots, P$, are obtained from a uniform distribution over the interval $[-1, 1]$, $c_0=0$, and the dynamic range of the ADC is $[-0.5, 0.5]$. For the recovery algorithm, each element of $\boldsymbol{\epsilon}$ are constrained to $\{-7,-6,\cdots,6,7\}$. The results for different $P$ and $N$ are shown in Fig. \ref{SimuPBL}. Fig. \ref{SimuPBL}(a) shows the theoretical results of Theorem \ref{theorem2}. The black area includes all cases where $H(N) < P + 1$, indicating that the model (\ref{periodicModel}) is unidentifiable. In contrast, the white area includes all cases where $H(N) \geq P + 1$, indicating that the model (\ref{periodicModel}) is identifiable. The critical red line represents the boundary between these two regions. The probabilities of successful recovery by the algorithm are shown in Fig. \ref{SimuPBL}(b). For cases where $H(N) < P + 1$, the probabilities are almost zero, whereas for cases where $H(N) \geq P + 1$, the probabilities are equal to $1$. These results validate Theorem \ref{theorem2}. The blue line denotes the critical line of the oversampling factor that ensures the identifiability of the model for any $N$, as introduced in Corollary \ref{PBScorollary3}. It can be seen in Fig. \ref{SimuPBL}(b) that when the oversampling factor is larger than $3\left(1+\frac{1}{P}\right)$, the probabilities of successful recovery are equal to $1$, which validates our conclusion.

\begin{figure}[htb!]
    \centering
    \includegraphics[width=0.45\textwidth]{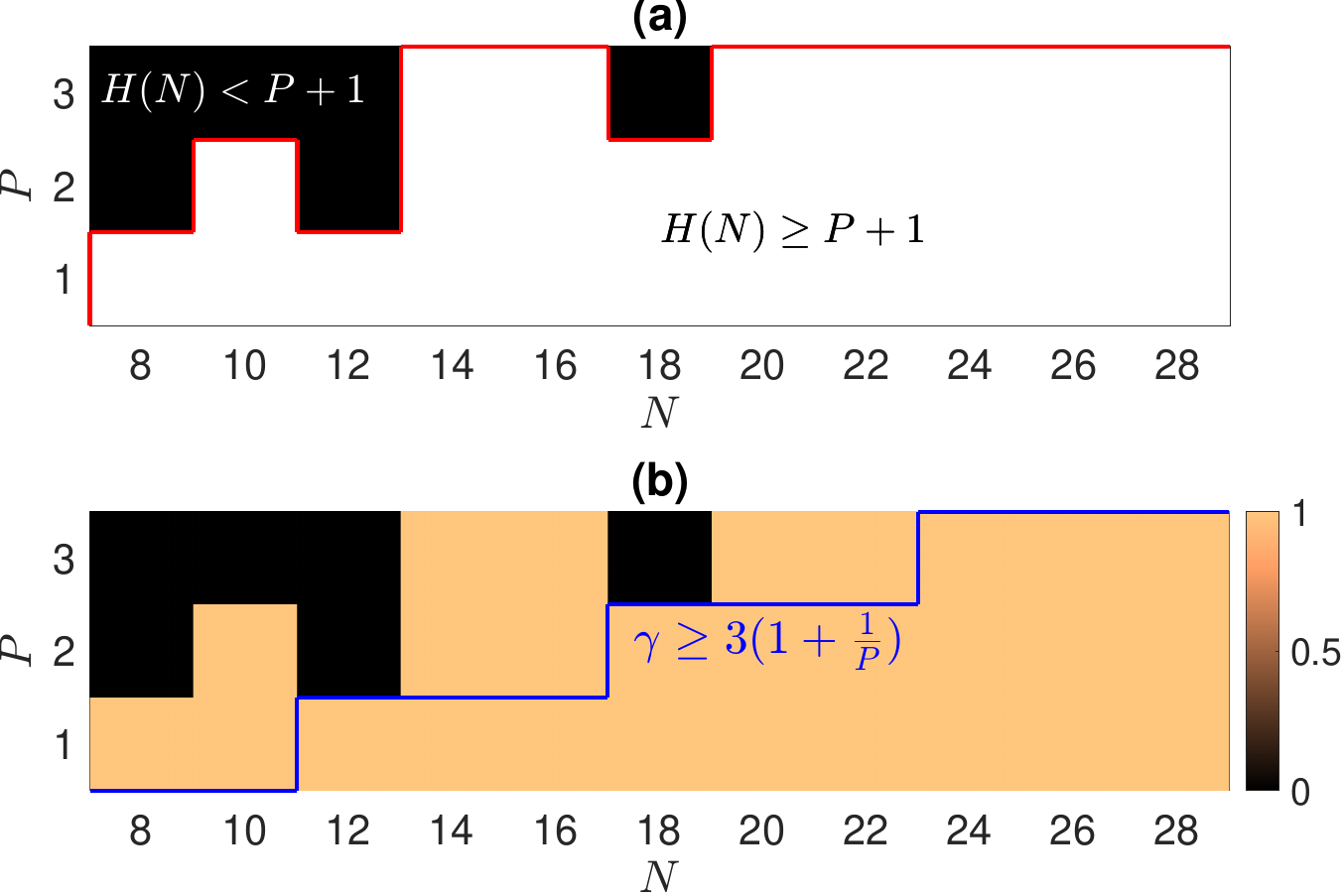}
    \caption{(a) The black part is the region that $H(N)<P+1$, indicating that the model (\ref{periodicModel}) is unidentifiable, and the white part is the region that $H(N)\geq P+1$, indicating that the model (\ref{periodicModel}) is identifiable. (b) The probabilities of successful recovery of PBL signals from modulo samples over $300$ Monte Carlo trials for different $P$ and $N$ of the proposed recovery algorithm.}
    \label{SimuPBL}
\end{figure}

\subsection{Proof of Theorem \ref{theorem3}}\label{ProofPBSTh3}
First, we transform the proposition that the PBL signal under MS (\ref{periodicModel}) is identifiable up to a constant factor into a proposition related to polynomials with integer coefficients, which is introduced in Proposition \ref{EquiCondiPBS}. In addition, to make the proof of Theorem \ref{theorem3} more brief, we prove Lemma \ref{PBSlemma1} in advance. Let $f(x)$ be a polynomial with rational coefficients of degree $N$ and let $\mathcal{R}_f$ be the set that contains all roots of $f(x)$. Lemma \ref{PBSlemma1} shows that $f(x) = A(x^N-1)$ holds for all $x\in\mathbb{C}$ if and only if $\forall d\in\mathcal{D}, \exists x\in\mathcal{R}_f$ such that $\Phi_d(x) = 0$, where $A$ is any nonzero number in $\mathbb{Q}$. Lemma \ref{PBSlemma1} is similar to Lemma \ref{lemmaT2_1}, and the only difference is that Lemma \ref{PBSlemma1} considers the rational number field $\mathbb{Q}$ while Lemma \ref{lemmaT2_1} considers the Gaussian rational number field $\mathbb{Q}[{\rm j}]$.
\begin{proposition}\label{EquiCondiPBS}
The samples of the PBL signal, i.e., $\mathbf{y}$, are uniquely identifiable up to an integer constant vector from modulo samples $\mathbf{z}$ if and only if there is no integer polynomial $f(x)$ of maximum degree $N-1$ not in the set $\{A\sum_{n=0}^{N-1}x^n\mid A\in\mathbb{Z}\}$ such that $f({\rm e}^{{\rm j}\frac{2\pi n}{N}}) = 0$ for all $n\in\mathcal{P}$.
\end{proposition}
\begin{proof}
    First, we have that $\mathbf{s}_{1:N-1}$ can be uniquely recovered from $\mathbf{z}$ if and only if there is no integer polynomial $f(x)$ of maximum degree $N-1$ not in $\{A\sum_{n=0}^{N-1}x^n\mid A\in\mathbb{Z}\}$ such that $f({\rm e}^{{\rm j}\frac{2\pi n}{N}}) = 0$ for all $n\in\mathcal{P}$. The proof is similar to that of Proposition \ref{EquiCondi} and is omitted here. Moreover, it can be easily shown that the unique recovery of $\mathbf{s}_{1:N-1}$ from $\mathbf{z}$ is equivalent to unfolded samples $\mathbf{y}$ in (\ref{periodicModel}) being uniquely identifiable from $\mathbf{z}$ up to an integer constant vector.
\end{proof}
\begin{lemma}\label{PBSlemma1}
For an integer coefficient polynomial $f(x)$ with degree $N$, the equation $f(x) = A(x^N-1)$, where $A$ is any nonzero number in $\mathbb{Q}$, holds for all $x\in\mathbb{C}$ if and only if $\forall d\in\mathcal{D}, \exists x\in\mathcal{R}_f$ such that $\Phi_d(x) = 0$.
\end{lemma}
\begin{proof}
    The proof is similar to that of Lemma \ref{lemmaT2_1}. The only difference is that the Gaussian rational number field $\mathbb{Q}[{\rm j}]$ is replaced by the rational number field $\mathbb{Q}$.
\end{proof}

The proof of Theorem \ref{theorem3} is provided below, which is based on Proposition \ref{EquiCondiPBS} and Lemma \ref{PBSlemma1}.

\textbf{\emph{Proof of Theorem \ref{theorem3}:}} 
According to Proposition \ref{EquiCondiPBS}, proving Theorem \ref{theorem3} is equivalent to proving the following statement: \emph{There is no integer polynomial $f(x)$ of maximum degree $N-1$ not in the set $\{A\sum_{n=0}^{N-1}x^n\mid A\in\mathbb{Z}\}$ such that $f({\rm e}^{{\rm j}\frac{2\pi n}{N}}) = 0$ for all $n\in\mathcal{P}$ if and only if $\forall d\in\mathcal{D}\setminus\{1\}, \exists n\in\mathcal{P}$ such that ${\rm e}^{{\rm j}\frac{2\pi n}{N}}$ is a root of $\Phi_{d}(x)$}.

We first prove the necessary part by proving the contrapositive. Let $i$ be an integer belonging to the set $\mathcal{D}\setminus\{1\}$, such that ${\rm e}^{{\rm j}\frac{2\pi n}{N}}$ is not a root of $\Phi_{i}(x)$ for any $n\in\mathcal{P}$. Below, we prove that there exists an integer polynomial $f(x)$ of maximum degree $N-1$ not in the set $\{A\sum_{n=0}^{N-1}x^n\mid A\in\mathbb{Z}\}$ such that $f({\rm e}^{{\rm j}\frac{2\pi n}{N}}) = 0$ for all $n\in\mathcal{P}$. The two cases, where $N$ is a prime number and where $N$ is a composite number, are investigated separately. When $N$ is a prime number, the set of all positive integers that divide $N$ is $\mathcal{D} = \{1, N\}$. Hence, we have $i=N$ and $\Phi_i(x) = \Phi_N(x) = x^{N-1} + x^{N-2} + \cdots + 1$. Since ${\rm e}^{{\rm j} \frac{2\pi n}{N}}, n \in \mathcal{P}$ are not roots of $\Phi_N(x)$, whereas $\{{\rm e}^{{\rm j} \frac{2\pi n}{N}}\}_{n=1}^{N-1}$ are roots of $\Phi_N(x)$, it follows that $\mathcal{P}$ (\ref{Pdefine}) can only be the empty set. Therefore, any integer polynomial of maximum degree $N-1$ satisfies $f({\rm e}^{{\rm j}\frac{2\pi n}{N}}) = 0$ for all $n\in\mathcal{P}$. In the following, we consider the case that $N$ is a composite number. Let $M(x)$ be a polynomial defined as  
\begin{align}\label{defMxPBS}
    M(x) \triangleq \prod_{d\in\mathcal{D},d\neq i}\Phi_{d}(x).
\end{align}
$M(x)$ is a monic integer polynomial because all cyclotomic polynomials are monic integer polynomials according to Property \ref{MonI}. Furthermore, since $i\neq 1$, we have $\Phi_i(x)\neq x-1$. Additionally, based on Property \ref{FacC}, it follows $x^N-1 = \prod_{d\in\mathcal{D}}\Phi_{d}(x) = (x-1)(\sum_{n=0}^{N-1}x^n)$. Therefore, $M(x)\neq \sum_{n=0}^{N-1}x^n$. Consequently, $M(x)$ is an integer polynomial of maximum degree $N-1$ that does not belong to the set $\{A\sum_{n=0}^{N-1}x^n\mid A\in\mathbb{Z}\}$. Below, we prove $M({\rm e}^{{\rm j}\frac{2\pi n}{N}}) = 0$ for all $n\in\mathcal{P}$. Since $x^N-1=\prod_{d\in\mathcal{D}}\Phi_{d}(x) = M(x)\Phi_i(x)$ and $\{{\rm e}^{{\rm j}\frac{2\pi n}{N}}\}_{n=0}^{N-1}$ are roots of $x^N-1=0$, ${\rm e}^{{\rm j}\frac{2\pi n}{N}}$ is a root of either $M(x)$ or $\Phi_i(x)$ for all $n\in\mathcal{P}$. Because ${\rm e}^{{\rm j}\frac{2\pi n}{N}}$ is not a root of $\Phi_{i}(x)$ for all $n\in\mathcal{P}$, ${\rm e}^{{\rm j}\frac{2\pi n}{N}}$ must be a root of $M(x)$ for all $n\in\mathcal{P}$, and the proof is completed.

Now, we prove the sufficient part by contradiction. Assume there exists an integer polynomial $f(x)$ of maximum degree $N-1$ not in $\{A\sum_{n=0}^{N-1}x^n\mid A\in\mathbb{Z}\}$ satisfying $f({\rm e}^{{\rm j}\frac{2\pi n}{N}}) = 0$ for all $n\in\mathcal{P}$. Let $g(x)\triangleq x^{\alpha}(x-1)f(x)$ where $\alpha = N - {\rm deg}(f(x))-1$, which is a polynomial with integer coefficients of degree $N$. Since $f( {\rm e}^{{\rm j} \frac{2\pi n}{N}}) = 0$ for all $n \in \mathcal{P}$, and for every $d \in \mathcal{D} \setminus \{1\}$, there exists $n \in \mathcal{P}$ such that ${\rm e}^{{\rm j} \frac{2\pi n}{N}}$ is a root of $\Phi_d(x)$, we can conclude that for all $d \in \mathcal{D} \setminus \{1\}$, there exists $x \in \mathcal{R}_g$ such that $\Phi_d(x) = 0$. Here, $\mathcal{R}_g$ denotes the set of roots of $g(x)$. Additionally, since $x - 1$ is a factor of $g(x)$, $\exists x \in \mathcal{R}_g$ such that $\Phi_1(x) = 0$. Thus, according to Lemma \ref{PBSlemma1}, we have $g(x) = A(x^N - 1)$, where $A$ is any nonzero integer. Therefore, we have $x^{\alpha} f(x) = \frac{g(x)}{x-1} = \frac{A(x^N - 1)}{x - 1} = A \sum_{n=0}^{N-1} x^n$. This implies that $\alpha = 0$ and $f(x) = A \sum_{n=0}^{N-1} x^n$, which contradicts the assumption.
\subsection{Proof of Lemma \ref{PBSlemma3}}\label{ProofTh3}
According to the definition of cyclotomic polynomials in Definition \ref{CycP}, $\frac{dn}{N}\in\mathbb{Z}$ and ${\rm gcd}(d,\frac{dn}{N})=1$ is equivalent to ${\rm e}^{{\rm j}\frac{2\pi n}{N}}$ is a root of $\Phi_{d}(x)$ as ${\rm e}^{{\rm j}\frac{2\pi n}{N}} = {\rm e}^{{\rm j}\frac{2\pi n d/N}{d}}$. Therefore, proving Lemma \ref{PBSlemma3} is equivalent to proving the following statement: \emph{$\forall d\in\mathcal{D}\setminus\{1\}, \exists n\in\mathcal{P}$ such that $\frac{dn}{N}\in\mathbb{Z}$ and ${\rm gcd}(d,\frac{dn}{N})=1$ if and only if $H(N)\geq P+1$.}

First, we prove the sufficient part. According to the definition of $h(d)$ (\ref{definehn}), we discuss $d$ in four different cases. For Case $1$ where $2\nmid d$, let $n = \frac{(d+1)N}{2d}$. As $d\mid N$ and $d$ is an odd number, $n$ is an integer. In addition, according to the definition of $H(N)$ (\ref{defineHN}), we have $\frac{N}{h(d)}\geq H(N)$. As $H(N)\geq P+1$ and $h(d) = 2+\frac{2}{d-1}$, we have $\frac{(d-1)N}{2d}\geq P+1$. Furthermore, we have $n = \frac{(d+1)N}{2d}\geq \frac{(d-1)N}{2d}\geq P+1$ and $n = \frac{(d+1)N}{2d} = N - \frac{(d-1)N}{2d}\leq N- P-1$. Since set $\mathcal{P}$ contains integers from $P+1$ to $N-P-1$ according to (\ref{Pdefine}), we have $n\in\mathcal{P}$. Finally, the greatest common divisor of $d$ and $\frac{dn}{N}$ is ${\rm gcd}(d,\frac{dn}{N}) = {\rm gcd}(d,\frac{d+1}{2}) = 1$ as ${\rm gcd}(d,d+1) = 1$. Similarly, for Case $2$ where $d=2$, Case $3$ where $2\mid d, d\neq 2\text{ and } 4\nmid d$, and Case $4$ where $4\mid d$, let $n = \frac{N}{2}$, $n = \frac{(d+4)N}{2d}$, and $n = \frac{(d+2)N}{2d}$, respectively. First, it can be easily proved that $n$ are integers for all three cases. Furthermore, according to $H(N)\geq P+1$ and $\frac{N}{h(d)}\geq H(N)$ where $h(d)=2$ in Case $2$, $h(d) = 2+\frac{8}{d-4}$ in Case $3$, and $h(d)=2+\frac{4}{d-2}$ in Case $4$, we have $\frac{N}{2}\geq P+1$ in Case $2$, $N/(2+\frac{8}{d-4})\geq P+1$ in Case $3$, and $N/(2+\frac{4}{d-2})\geq P+1$ in Case $4$. Thus, it can be verified that $n\in\mathcal{P}$ for $n = \frac{N}{2}$ in Case $2$, $n = \frac{(d+4)N}{2d}$ in Case $3$, and $n = \frac{(d+2)N}{2d}$ in Case $4$. Finally, we have ${\rm gcd}(d,\frac{dn}{N}) = {\rm gcd}(2,1) = 1$ for Case $2$ in which $d=2$ and $n=\frac{N}{2}$. For Case $3$ in which $2\mid d$, $d\neq 2$, $4\nmid d$ and $n = \frac{(d+4)N}{2d}$, we have ${\rm gcd}(d,\frac{dn}{N}) = {\rm gcd}(d, \frac{d+4}{2}) = 1$ as ${\rm gcd}(d, d+4) = 2$. For Case $4$ in which $4\mid d$ and $n = \frac{(d+2)N}{2d}$, we have ${\rm gcd}(d,\frac{dn}{N}) = {\rm gcd}(d, \frac{d+2}{2}) = 1$ as ${\rm gcd}(d,d+2) = 2$.

We now prove the necessary part by proving the contrapositive, i.e., if $H(N)<P+1$, then $\exists d_1\in\mathcal{D}\setminus\{1\}$, there does not exist an $n\in\mathcal{P}$ satisfying $\frac{d_1n}{N}\in\mathbb{Z}$ and ${\rm gcd}(d_1,\frac{d_1n}{N}) = 1$. According to the definition of $H(N)$ (\ref{defineHN}), $H(N)<P+1$ is equivalent to $\exists d_2\in\mathcal{D}\setminus\{1\}, \frac{N}{h(d_2)}< P+1$. Therefore, we need to prove that $\exists d_2\in\mathcal{D}\setminus\{1\}$ such that $\frac{N}{h(d_2)}< P+1$, then $\exists d_1\in\mathcal{D}\setminus\{1\}$, there does not exist an $n\in\mathcal{P}$ satisfying $\frac{d_1n}{N}\in\mathbb{Z}$ and ${\rm gcd}(d_1,\frac{d_1n}{N}) = 1$. We discuss $d_2$ in four different cases according to the definition of $h(d)$ (\ref{definehn}). For Case $1$ in which $2\nmid d_2$, $h(d_2) = 2+\frac{2}{d_2-1}$. Then, $\frac{N}{h(d_2)}< P+1$ is equivalent to $\frac{d_2(P+1)}{N}>\frac{d_2-1}{2}$. According to the inequality $\frac{d_2(P+1)}{N}>\frac{d_2-1}{2}$ which is also equivalent to $\frac{d_2(N-P-1)}{N} < \frac{d_2+1}{2}$ and the definition of $\mathcal{P}$, i.e., $ \mathcal{P} = \{P+1,P+2,\cdots,N-P-1\}$ (\ref{Pdefine}), we have $\frac{d_2-1}{2}<\frac{d_2n}{N}<\frac{d_2+1}{2}$ for all $n\in\mathcal{P}$. Since $d_2$ is an odd number, there is no integer in the range $(\frac{d_2-1}{2}, \frac{d_2+1}{2})$. Let $d_1 = d_2$, we have that there does not exist an $n\in\mathcal{P}$ satisfying $\frac{d_1n}{N}\in\mathbb{Z}$. For Case $2$ in which $d_2 = 2$, we have $h(d_2) = 2$. Then, $\frac{N}{h(d_2)}< P+1$ is equivalent to $\frac{2(P+1)}{N}>1$ (or equivalently, $\frac{2(N -P - 1)}{N} < 1$). Similar to Case $1$, it can be derived that $1  < \frac{2(P+1)}{N} < \frac{d_2n}{N} = \frac{2 n}{N} < \frac{2(N - P - 1)}{N} < 1$ for all $n\in\mathcal{P}$. Let $d_1 = d_2$, we also have that there does not exist an $n\in\mathcal{P}$ satisfying $\frac{d_1n}{N}\in\mathbb{Z}$. For Case $3$ where $2\mid d_2$, $d_2\neq 2$, $4\nmid d_2$ and $h(d_2) = 2+\frac{8}{d_2-4}$, $\frac{N}{h(d_2)}< P+1$ is equivalent to $\frac{d_2(P+1)}{N}>\frac{d_2-4}{2}$ (or equivalently, $\frac{d_2(N - P - 1)}{N} < \frac{d_2+4}{2}$). Similarly, we have $\frac{d_2-4}{2}<\frac{d_2n}{N}<\frac{d_2+4}{2}$ for all $n\in\mathcal{P}$. Let $d_1 = d_2$ and there are three integers in the range $(\frac{d_1-4}{2},\frac{d_1+4}{2})$ which are $\frac{d_1}{2}-1,\frac{d_1}{2},\frac{d_1}{2}+1$. As $\frac{d_1}{2}-1$ and $\frac{d_1}{2}+1$ are even integers, therefore $\frac{d_1}{2}-1,\frac{d_1}{2}+1$ are not coprime to $d_1$. In addition, since $d_1 > 2$, $\frac{d_1}{2}$ is not coprime to $d_1$. Therefore, there does not exist an $n\in\mathcal{P}$ satisfying $\frac{d_1n}{N}\in\mathbb{Z}$ and ${\rm gcd}(d_1,\frac{d_1n}{N}) =  1$. For Case $4$ where $4 \mid d_2$ and $h(d_2) = 2+\frac{4}{d_2-2}$, $\frac{N}{h(d_2)} < P+1$ is equivalent to $\frac{d_2(P+1)}{N}>\frac{d_2-2}{2}$ (or equivalently, $\frac{d_2(N - P - 1)}{N} < \frac{d_2+2}{2}$). Therefore, we have $\frac{d_2-2}{2}<\frac{d_2n}{N}<\frac{d_2+2}{2}$ for all $n\in\mathcal{P}$. Let $d_1 = d_2$, $\frac{d_1}{2}$ is the only integer in the range $(\frac{d_1-2}{2},\frac{d_1+2}{2})$ which is also an even number since $4\mid d_1$. Thus, $\frac{d_1}{2}$ is not coprime to $d_1$. Therefore, there does not exist an $n\in\mathcal{P}$ satisfying $\frac{d_1n}{N}\in\mathbb{Z}$ and ${\rm gcd}(d_1,\frac{d_1n}{N}) = 1$.

\section{Conclusion}
This paper investigates the identifiability of the modulo-DFT sensing model. We present a necessary and sufficient condition for the unique recovery of the original signal from modulo measurements. Furthermore, we demonstrate that when any original signal has at least two elements (including the first one) being zero and the number of measurements is a prime number, the modulo-DFT sensing model is identifiable. Additionally, we examine the identifiability of PBL signals. We introduce the necessary and sufficient condition regarding the number of measurements for the unique identification of the PBL signals. Moreover, we establish that the PBL signal can be uniquely identified when the oversampling rate exceeds $3(1+1/P)$ from the modulo samples, where $P$ is the number of harmonics including the fundamental component in the positive frequency part. Finally, we propose an algorithm to estimate the original signal by solving integer linear equations, and we perform simulations to verify our conclusions.

\section{Appendix}
\subsection{Proof of Property \ref{LemGauRat}}\label{AppenA}
Before proving Property \ref{LemGauRat}, several definitions and lemmas are introduced first. First, Definition \ref{DefGauPri} introduces the definition of Gaussian primes. Similarly to prime numbers, a Gaussian prime is a Gaussian integer that is irreducible over $\mathbb{Z}[{\rm j}]$. The units of $\mathbb{Z}[{\rm j}]$ are precisely the Gaussian integers with norm $1$, that are, $1$, $-1$, $-{\rm j}$ and ${\rm j}$. In addition, given a Gaussian prime $p$, the prime ideal is defined in Definition \ref{DefPriIde}, which is a subset of $\mathbb{Z}[{\rm j}]$ that contains all multiples of $p$. For a prime number $p>2$, Lemma \ref{GaussPri} reveals that if $p$ is congruent to $3$ modulo 4, then $p$ is also a Gaussian prime; otherwise, $p$ is the product of a Gaussian prime by its conjugate.
Furthermore, for a prime number $p$, Lemma \ref{BinoPri} shows that $p \mid \binom{p}{k},k=1,2,\cdots,p-1$, where $\binom{p}{k}$ are binomial coefficients.

\begin{definition}(Gaussian prime)\label{DefGauPri}
    Let $p\in\mathbb{Z}[{\rm j}]$ be a Gaussian integer. $p$ is a Gaussian prime if and only if it is irreducible over $\mathbb{Z}[{\rm j}]$, that is, it is not the product of two non-units in $\mathbb{Z}[{\rm j}]$.
\end{definition}

\begin{definition}\cite{WikiPrId}(Prime ideal)\label{DefPriIde}
    Let $p\in\mathbb{Z}[{\rm j}]$ be a Gaussian prime. The prime ideal for $p$ denoted as $\mathbb{P}$ is $\mathbb{P}\triangleq\{pi \mid i\in\mathbb{Z}[{\rm j}]\}$, which is the set that contains all multiples of $p$.
\end{definition}

\begin{lemma}\cite{NeilDonaldson}\label{GaussPri}
    Let $p\in\mathbb{Z}$ be a prime number and $p>2$. If $p$ is congruent to $3$ modulo $4$, then $p$ is a Gaussian prime. If $p$ is congruent to $1$ modulo $4$, then $p$ is the product of a Gaussian prime by its conjugate.
\end{lemma}

\begin{lemma}\cite{fine1947binomial}\label{BinoPri}
    Let $p\in\mathbb{Z}$ be a prime number. Then, we have $p \mid \binom{p}{k}, k=1,2,\cdots,p-1$,
where $\binom{p}{k}$ are binomial coefficients.
\end{lemma}

Let $f(x)$ be a polynomial with Gaussian integer coefficients. Lemma \ref{GaussLemma} introduces the Gauss's lemma which shows that if $f(x)$ is reducible over $\mathbb{Q}[{\rm j}]$, then $f(x)$ is also reducible over $\mathbb{Z}[{\rm j}]$. In other words, we have that if $f(x)$ is irreducible over $\mathbb{Z}[{\rm j}]$, then $f(x)$ is also irreducible over $\mathbb{Q}[{\rm j}]$. In addition, Lemma \ref{EisenCrit} introduces the Eisenstein's criterion, which can be used to determine the irreducibility of $f(x)$ over $\mathbb{Z}$[{\rm j}]. Furthermore, Lemma \ref{PrimeIrrLe} shows that $f(x) = x^{p-1}+x^{p-2} + \cdots+1$ is irreducible over $\mathbb{Z}[{\rm j}]$, where $p$ is a prime number. 

\begin{lemma}\cite[Chapter 9, Proposition 5]{dummit2004abstract}(Gauss's lemma)\label{GaussLemma}
    Let $f(x)$ be a polynomial with Gaussian integer coefficients. if $f(x)$ is reducible over $\mathbb{Q}[{\rm j}]$ then $f(x)$ is reducible over $\mathbb{Z}[{\rm j}]$.
\end{lemma}

\begin{lemma}\cite[Chapter $9$, Proposition $13$]{dummit2004abstract}(Eisenstein's criterion)\label{EisenCrit}
Let $\mathbb{P}$ be a prime ideal of $\mathbb{Z}[{\rm j}]$ and let $f(x)=x^n+a_{n-1} x^{n-1}+\cdots+a_1 x+a_0$ be a polynomial over $\mathbb{Z}[{\rm j}]$ ($n \geq 1$). Suppose $a_{n-1}, \ldots, a_1, a_0$ are all elements of $\mathbb{P}$ and suppose $a_0$ is not an element of $\mathbb{P}^2$ where $\mathbb{P}^2 = \{ab\mid a\in\mathbb{P}, b\in \mathbb{P}\}$, then $f(x)$ is irreducible over $\mathbb{Z}[{\rm j}]$.
\end{lemma}

\begin{lemma}\label{PrimeIrrLe}
    Let $f(x) = x^{p-1}+x^{p-2}+\cdots+1$, where $p\in\mathbb{Z}$ is a prime number. $f(x)$ is irreducible over $\mathbb{Z}[{\rm j}]$.
\end{lemma}
\begin{proof}
    First, when $p=2$ we have $f(x) = x+1$. It is obvious that $f(x)$ is irreducible over $\mathbb{Z}[{\rm j}]$. 
    
    In the following, we consider the case when $p>2$. Note that $f(x) = \frac{x^p-1}{x-1}$. By substituting $x$ with $x+1$ in $f(x)$, we have $f(x+1) = \frac{(x+1)^p-1}{x} = \sum_{i=1}^p \binom{p}{i}x^{i-1}$, where $\binom{p}{i}$ are binomial coefficients. In addition, let $a_{i-1} = \binom{p}{i}, i=1,2,\cdots,p-1$, we have $f(x+1) = x^{p-1} + \sum_{i=0}^{p-2} a_{i}x^{i}$. It is worth noting that if $f(x+1)$ is irreducible over $\mathbb{Z}[{\rm j}]$ then $f(x)$ is also irreducible over $\mathbb{Z}[{\rm j}]$. As $p$ is a prime number greater than $2$, it follows that $p$ is congruent to either $1$ or $3$ modulo $4$. We first consider the case that $p$ is congruent to $3$ modulo $4$. According to Lemma \ref{GaussPri}, $p$ is a Gaussian prime. Let $\mathbb{P}_1 \triangleq \{pi \mid i\in\mathbb{Z}[{\rm j}]\}$ be the prime ideal generated by $p$ over $\mathbb{Z}[{\rm j}]$. According to Lemma \ref{BinoPri}, we have $a_{0},a_{1},\cdots,a_{p-2}\in\mathbb{P}_1$. In addition, it is obvious that $a_{0}=p\notin\mathbb{P}_1^2$, where $\mathbb{P}_1^2 = \{ab\mid a\in\mathbb{P}_1, b\in \mathbb{P}_1\}$. Therefore, according to Eisenstein’s criterion, i.e., Lemma \ref{EisenCrit}, $f(x+1)$ is irreducible over $\mathbb{Z}[{\rm j}]$, which implies that $f(x)$ is also irreducible over $\mathbb{Z}[{\rm j}]$. We now consider the case that $p$ is congruent to $1$ modulo $4$. According to Lemma \ref{GaussPri}, $p$ can be represented as $p=p_1p_1^*$, where $p_1$ and $p_1^*$ are Gaussian primes. Let $p_1 = m+n{\rm j}$ where $m,n\in\mathbb{Z}$. Since $p$ is a prime number that cannot be expressed as the square of an integer, we have $m\neq 0$ and $n\neq 0$ \footnote{When $m=0$, we have $p=p_1p_1^* = n{\rm j}(-n{\rm j}) = n^2$.}. Let $\mathbb{P}_2\triangleq \{p_1i \mid i\in\mathbb{Z}[{\rm j}]\}$ be the prime ideal generated by $p_1$ over $\mathbb{Z}[{\rm j}]$. Note that $\mathbb{P}_2$ contains all multiples of $p$ as $p=p_1p_1^*$. Thus, according to Lemma \ref{BinoPri}, we have $a_{0},a_{1},\cdots,a_{p-2}\in\mathbb{P}_2$. Below, we prove $a_0\notin\mathbb{P}^2_2$ by contradiction. Each nonzero element in $\mathbb{P}_2^2$ can be represented as $p_1^2q$, where $q\in \mathbb{Z}[{\rm j}]$. Assuming $a_0\in \mathbb{P}_2^2$, we have $p_1p_1^* = p_1^2q$ for some $q\in \mathbb{Z}[{\rm j}]$. Therefore, we have $p_1^* = p_1q$. Since $p_1^* = m - n{\rm j}$ is a Gaussian prime, $q$ can only be the units of $\mathbb{Z}[{\rm j}]$, i.e., $q\in\{1,{\rm j},-1,-{\rm j}\}$. In addition, for $m\neq 0$ and $n\neq 0$, $q$ is either equal to ${\rm j}$ or equal to $-{\rm j}$ \footnote{If $q=1$, we have $n=0$ since $p_1^* = p_1q$.}. Thus, we have $q=1$ or $q=-1$, which implies $|m| = |n|$. Therefore, $p=p_1p_1^* = (m+n{\rm j})(m-n{\rm j}) = m^2+n^2 = 2m^2 = 2n^2$, which is an even number. This contradicts the fact that $p$ is a prime number greater than $2$. Consequently, we can infer that $a_0\notin \mathbb{P}_2^2$. Similarly, according to Eisenstein’s criterion, i.e., Lemma \ref{EisenCrit}, we deduce that $f(x+1)$ is irreducible over $\mathbb{Z}[{\rm j}]$ when $p>2$, which implies that $f(x)$ is also irreducible over $\mathbb{Z}[{\rm j}]$. 
\end{proof}
\textbf{\emph{Proof of Property \ref{LemGauRat}:}} By Lemma \ref{PrimeIrrLe}, $f(x) = x^{p-1}+x^{p-2}+\cdots+1$ is irreducible over $\mathbb{Z}[{\rm j}]$. Furthermore, according to Gauss's lemma in Lemma \ref{GaussLemma}, $f(x)$ is also irreducible over $\mathbb{Q}[{\rm j}]$.

\subsection{Proof of Property \ref{LemPowGau}}\label{AppenB}
Before proving Property \ref{LemPowGau}, we first introduce Lemma \ref{BinoPri2} and Lemma \ref{PowIrrLe}. Let $m$ be a positive integer. Lemma \ref{BinoPri2} shows that $2\mid \binom{2^m}{k}$ for any integer $k$ satisfying $1\leq k < 2^m$. Furthermore, Lemma \ref{PowIrrLe} reveals that $f(x) = x^{2^m} - {\rm j}$ or $f(x) = x^{2^m} + {\rm j}$ is irreducible over $\mathbb{Z}[{\rm j}]$.

\begin{lemma}\cite{fine1947binomial}\label{BinoPri2}
    Let $m$ and $k$ be positive integers. Then, we have $2 \mid \binom{2^m}{k}, k=1,2,\cdots,2^m-1$.
\end{lemma}

\begin{lemma}\label{PowIrrLe}
    Let $f(x) = x^{2^m}-{\rm j}$ or $f(x) = x^{2^m}+{\rm j}$ where $m$ is a positive integer. $f(x)$ is irreducible over $\mathbb{Z}[{\rm j}]$.
\end{lemma}
\begin{proof}
    We only provide the proof of the case that $f(x) = x^{2^m}-{\rm j}$. For the other case where $f(x) = x^{2^m}+{\rm j}$, the proof is similar and is omitted here. By substituting $x$ with $x+1$ in $f(x)$, we have $f(x+1) = (x+1)^{2^m}-{\rm j} = \sum_{i=1}^{2^m} \binom{2^m}{i}x^{i} + 1-{\rm j}$. In addition, let $a_{i} = \binom{2^m}{i}, i=1,2,\cdots,2^m-1$ and $a_0 = 1-{\rm j}$, we have $f(x+1) = x^{2^m} + \sum_{i=0}^{2^m-1} a_{i}x^{i}$. Note that if $f(x+1)$ is irreducible over $\mathbb{Z}[{\rm j}]$ then $f(x)$ is also irreducible over $\mathbb{Z}[{\rm j}]$. Let $p = 1-{\rm j}$, which is a Gaussian prime, and let $\mathbb{P}$ be the prime ideal generated by $p$ over $\mathbb{Z}[{\rm j}]$. According to Lemma \ref{BinoPri2}, we have $a_{1},\cdots,a_{2^m-1}\in\mathbb{P}$ as $2 = (1-{\rm j})(1+{\rm j}) = pp^*$. In addition, we also have $a_0\in\mathbb{P}$ as $a_0 = (1-{\rm j}) = p$. Moreover, it is obvious that $a_{0}=p\notin\mathbb{P}^2$, where $\mathbb{P}^2 = \{ab\mid a\in\mathbb{P}, b\in \mathbb{P}\}$. Thus, according to Eisenstein’s criterion, i.e., Lemma \ref{EisenCrit}, we have that $f(x+1)$ is irreducible over $\mathbb{Z}[{\rm j}]$, which implies that $f(x)$ is irreducible over $\mathbb{Z}[{\rm j}]$. 
\end{proof}
\textbf{\emph{Proof of Property \ref{LemPowGau}:}} According to Lemma \ref{PowIrrLe}, $x^{2^m} - {\rm j}$ and $ x^{2^m} + {\rm j}$ are irreducible over $\mathbb{Z}[{\rm j}]$. Furthermore, based on Gauss's lemma in Lemma \ref{GaussLemma}, $x^{2^m} - {\rm j}$ and $ x^{2^m} + {\rm j}$ are irreducible over $\mathbb{Q}[{\rm j}]$.
\subsection{Proof of Proposition \ref{CycPriP}}\label{AppenC}
Let $g(x)\in\mathbb{F}[x]$ be a polynomial of the minimum degree having a root at $\xi$. According to the polynomial remainder theorem in Proposition \ref{polyReTh}, we have $f(x) = g(x)q(x) + r(x)$, where $q(x)$ and $r(x)$ are polynomials over $\mathbb{F}$, and the degree of $r(x)$ is less than that of $g(x)$. Because $f(\xi) = g(\xi)q(\xi) + r(\xi) =0$ and $g(\xi)=0$, we have $r(\xi) = 0$. As $g(x)$ is a polynomial of the minimum degree having a root at $\xi$ and the degree of $r(x)$ is less than that of $g(x)$, we have $r(x)= 0$. Therefore, we have $f(x) = g(x)q(x)$. Since $f(x)$ is irreducible over $\mathbb{F}$, $q(x)$ is a constant. Consequently, $g(x)$ has the same degree as $f(x)$. 
\subsection{Cyclotomic polynomial}\label{AppenD}
Let $\xi_{n,k}\triangleq {\rm e}^{{\rm j}\frac{2\pi k}{n}}$ be an $n$th root of unity, where $k$ and $n$ are positive integers and $k \leq n$. In Definition \ref{PriRU}, we show that $\xi_{n,k}$ is a primitive $n$th root of unity if and only if ${\rm gcd}(k,n)=1$, where ${\rm gcd}(k,n)$ denotes the greatest common divisor of $k$ and $n$. In addition, in Definition \ref{CycP}, we show that the $n$th cyclotomic polynomial is a polynomial with its roots being primitive $n$th roots of unity. Let $\Phi_n(x)$ be the $n$th cyclotomic polynomial. Property \ref{FacC} shows that the product of $\Phi_d(x)$ for all $d\in\mathcal{D}$ equals to $x^N-1$, where $\mathcal{D}$ is the set that contains all positive integers that divide $N$ (\ref{DefD_N}). In addition, Property \ref{MonI} shows that $\Phi_n(x)$ is a monic polynomial with integer coefficients. Finally, Property \ref{irre_poly} shows that $\Phi_n(x)$ is irreducible over the rational number field $\mathbb{Q}$. 

\begin{definition}\label{PriRU}(Primitive $n$th root of unity)
Let $\xi_{n,k} \triangleq{\rm e}^{{\rm j} \frac{2\pi k}{n}}$ which is an $n$th root of unity where $k$ and $n$ are both positive integers and $k\leq n$. $\xi_{n,k}$ is a primitive $n$th root of unity if and only if ${\rm gcd}(k,n)=1$.
\end{definition}

\begin{definition}\label{CycP}(Cyclotomic polynomial)
Let $\mathcal{T}$ be the set of primitive $n$th roots of unity, i.e., $\mathcal{T}\triangleq\{{\rm e}^{{\rm j} \frac{2 \pi k}{n}}\mid k\in\{1,2,\cdots,n\}\text{ and } {\rm gcd}(k,n)=1\}$. The $n$th cyclotomic polynomial denoted as $\Phi_n(x)$ is defined as 
\begin{align}
    \Phi_n(x) =\prod_{x_s \in \mathcal{T}}\left(x-x_s\right).
\end{align}
\end{definition}

\begin{property}\cite[Lemma 5.11]{aluffi2021algebra}\label{FacC}
Let $\mathcal{D}$ be the set that contains all positive integers that divide $N$, we have
\begin{align}
    x^N-1 = \prod_{d\in\mathcal{D}}\Phi_d(x).
\end{align}
\end{property}

\begin{property}\cite[Chapter 13, Lemma 40]{dummit2004abstract}\label{MonI}
Let $n$ be any positive integer. The $n$th cyclotomic polynomial $\Phi_n(x)$ is a monic polynomial with integer coefficients.
\end{property}

\begin{property}\cite[Proposition 5.14]{aluffi2021algebra}\label{irre_poly}
Let $n$ be any positive integer, the $n$th cyclotomic polynomial $\Phi_n(x)$ is irreducible over the rational number field $\mathbb{Q}$.
\end{property}
\bibliography{mybib}

\begin{thebibliography}{10}

\bibitem{zhang2018low}
J.~Zhang, L.~Dai, X.~Li, Y.~Liu, and L.~Hanzo, ``On low-resolution {ADC}s in
  practical 5{G} millimeter-wave massive {MIMO} systems,'' {\em IEEE Commun.
  Mag.}, vol.~56, no.~7, pp.~205--211, 2018.

\bibitem{zhang2019range}
R.~Zhang, C.~Li, J.~Li, and G.~Wang, ``Range estimation and range-doppler
  imaging using signed measurements in {LFMCW} radar,'' {\em IEEE Trans.
  Aerosp. Electron. Syst.}, vol.~55, no.~6, pp.~3531--3550, 2019.

\bibitem{bhandari2017unlimited}
A.~Bhandari, F.~Krahmer, and R.~Raskar, ``On unlimited sampling,'' in {\em
  Proc. Int. Conf. Sampling Theory Appl.}, pp.~31--35, 2017.

\bibitem{bhandari2020unlimited}
A.~Bhandari, F.~Krahmer, and R.~Raskar, ``On unlimited sampling and
  reconstruction,'' {\em IEEE Trans. Signal Process.}, vol.~69, pp.~3827--3839,
  2020.

\bibitem{bhandari2021unlimited}
A.~Bhandari, F.~Krahmer, and T.~Poskitt, ``Unlimited sampling from theory to
  practice: {F}ourier-{P}rony recovery and prototype {ADC},'' {\em IEEE Trans.
  Signal Process.}, vol.~70, pp.~1131--1141, 2021.

\bibitem{florescu2022surprising}
D.~Florescu, F.~Krahmer, and A.~Bhandari, ``The surprising benefits of
  hysteresis in unlimited sampling: {T}heory, algorithms and experiments,''
  {\em IEEE Trans. Signal Process.}, vol.~70, pp.~616--630, 2022.

\bibitem{ordentlich2018modulo}
O.~Ordentlich, G.~Tabak, P.~K. Hanumolu, A.~C. Singer, and G.~W. Wornell, ``A
  modulo-based architecture for analog-to-digital conversion,'' {\em IEEE J.
  Sel. Topics Signal Process.}, vol.~12, no.~5, pp.~825--840, 2018.

\bibitem{bhandari2019identifiability}
A.~Bhandari and F.~Krahmer, ``On identifiability in unlimited sampling,'' in
  {\em Proc. Int. Conf. Sampling Theory Appl.}, pp.~1--4, 2019.

\bibitem{romanov2019above}
E.~Romanov and O.~Ordentlich, ``Above the {N}yquist rate, modulo folding does
  not hurt,'' {\em IEEE Signal Process. Lett.}, vol.~26, no.~8, pp.~1167--1171,
  2019.

\bibitem{prasanna2020identifiability}
D.~Prasanna, C.~Sriram, and C.~R. Murthy, ``On the identifiability of sparse
  vectors from modulo compressed sensing measurements,'' {\em IEEE Signal
  Process. Lett.}, vol.~28, pp.~131--134, 2020.

\bibitem{shtendel2023unlimited}
G.~Shtendel, D.~Florescu, and A.~Bhandari, ``Unlimited sampling of bandpass
  signals: {Computational} demodulation via undersampling,'' {\em IEEE Trans.
  Signal Process.}, vol.~71, pp.~4134--4145, 2023.

\bibitem{mulleti2024modulo}
S.~Mulleti and Y.~C. Eldar, ``Modulo sampling of {FRI} signals,'' {\em IEEE
  Access}, vol.~12, pp.~60369--60384, 2024.

\bibitem{ordentlich2016integer}
O.~Ordentlich and U.~Erez, ``Integer-forcing source coding,'' {\em IEEE Trans.
  Inf. Theory}, vol.~63, no.~2, pp.~1253--1269, 2016.

\bibitem{romanov2021blind}
E.~Romanov and O.~Ordentlich, ``Blind unwrapping of modulo reduced {G}aussian
  vectors: {R}ecovering {MSB}s from {LSB}s,'' {\em IEEE Trans. Inf. Theory},
  vol.~67, no.~3, pp.~1897--1919, 2021.

\bibitem{weiss2022blind}
A.~Weiss, E.~Huang, O.~Ordentlich, and G.~W. Wornell, ``Blind modulo
  analog-to-digital conversion,'' {\em IEEE Trans. Signal Process.}, vol.~70,
  pp.~4586--4601, 2022.

\bibitem{weiss2022blind2}
A.~Weiss, E.~Huang, O.~Ordentlich, and G.~W. Wornell, ``Blind modulo
  analog-to-digital conversion of vector processes,'' in {\em Proc. IEEE Int.
  Conf. Acoust., Speech Signal Process.}, pp.~5617--5621, 2022.

\bibitem{cucuringu2020provably}
M.~Cucuringu and H.~Tyagi, ``Provably robust estimation of modulo 1 samples of
  a smooth function with applications to phase unwrapping,'' {\em J. Mach.
  Learn. Res.}, vol.~21, no.~32, 2020.

\bibitem{fanuel2021recovering}
M.~Fanuel and H.~Tyagi, ``Recovering h{\"o}lder smooth functions from noisy
  modulo samples,'' in {\em Proc. Asilomar Conf. on Signals, Syst., and
  Comput.}, pp.~857--861, IEEE, 2021.

\bibitem{fanuel2022denoising}
M.~Fanuel and H.~Tyagi, ``Denoising modulo samples: k-{NN} regression and
  tightness of {SDP} relaxation,'' {\em Inf. Inference: A J. IMA}, vol.~11,
  no.~2, pp.~637--677, 2022.

\bibitem{tyagi2022error}
H.~Tyagi, ``Error analysis for denoising smooth modulo signals on a graph,''
  {\em Appl. Comput. Harmon. Anal.}, vol.~57, pp.~151--184, 2022.

\bibitem{zhang2024line}
Q.~Zhang, J.~Zhu, F.~Qu, and D.~W. Soh, ``Line spectral estimation via
  unlimited sampling,'' {\em IEEE Trans. Aerosp. Electron. Syst.}, 2024.

\bibitem{cheng2023crb}
Y.~Cheng, J.~Karlsson, and J.~Li, ``{CRB} analysis for {Mod-ADC} with known
  folding-count,'' in {\em Proc. IEEE Veh. Technol. Conf.}, pp.~1--5, 2023.

\bibitem{eamaz2024uno}
A.~Eamaz, K.~V. Mishra, F.~Yeganegi, and M.~Soltanalian, ``{UNO}: {U}nlimited
  sampling meets one-bit quantization,'' {\em IEEE Trans. on signal Process.},
  vol.~72, pp.~997--1014, 2024.

\bibitem{candes2006robust}
E.~J. Cand{\`e}s, J.~Romberg, and T.~Tao, ``Robust uncertainty principles:
  {E}xact signal reconstruction from highly incomplete frequency information,''
  {\em IEEE Trans. Inf. Theory}, vol.~52, no.~2, pp.~489--509, 2006.

\bibitem{ma2014turbo}
J.~Ma, X.~Yuan, and L.~Ping, ``Turbo compressed sensing with partial {DFT}
  sensing matrix,'' {\em IEEE Signal Process. Lett.}, vol.~22, no.~2,
  pp.~158--161, 2014.

\bibitem{gardner}
R.~B. Gardner, ``Factorizations of polynomials over a field,'' {\em available
  at https://faculty.etsu.edu/gardnerr/4127/notes/IV-23.pdf}, 2024.

\bibitem{gallian2021contemporary}
J.~Gallian, {\em Contemporary abstract algebra}.
\newblock Chapman and Hall/CRC, 2021.

\bibitem{bolusani2024scip}
S.~Bolusani, M.~Besan{\c{c}}on, K.~Bestuzheva, A.~Chmiela, J.~Dion{\'\i}sio,
  T.~Donkiewicz, J.~van Doornmalen, L.~Eifler, M.~Ghannam, A.~Gleixner, {\em
  et~al.}, ``The {SCIP} optimization suite 9.0,'' {\em arXiv preprint
  arXiv:2402.17702}, 2024.

\bibitem{gurobi}
{Gurobi Optimization, LLC}, ``{Gurobi Optimizer Reference Manual},'' 2024.

\bibitem{WikiPrId}
Wikipedia, ``Prime ideal,'' {\em available at
  https://en.wikipedia.org/wiki/Prime\_ideal}, 2024.

\bibitem{NeilDonaldson}
D.~Neil, ``Gaussian integers and rings of algebraic integers,'' {\em available
  at https://www.math.uci.edu/~ndonalds/math180b/6gaussian.pdf}, 2021.

\bibitem{fine1947binomial}
N.~J. Fine, ``Binomial coefficients modulo a prime,'' {\em Amer. Math.
  Monthly}, vol.~54, no.~10, pp.~589--592, 1947.

\bibitem{dummit2004abstract}
D.~S. Dummit and R.~M. Foote, {\em Abstract algebra}, vol.~3.
\newblock Wiley Hoboken, 2004.

\bibitem{aluffi2021algebra}
P.~Aluffi, {\em Algebra: chapter 0}, vol.~104.
\newblock American Mathematical Soc., 2021.

\end{thebibliography}
\bibliographystyle{ieeetr}
\end{document}